%% file: main.tex
\PassOptionsToPackage{destlabel}{hyperref}
\PassOptionsToPackage{table}{xcolor}
\documentclass[acmsmall, nonacm, screen]{acmart}
\makeatletter
\AtBeginDocument{%
  \setlength\labelsep{4pt}%
  \setlength{\@ACM@labelwidth}{6.5pt}%
  \setlength\leftmargini{\z@}%
  \addtolength\leftmargini{\parindent}%
  \addtolength\leftmargini{2\labelsep}%
  \addtolength\leftmargini{\@ACM@labelwidth}%
  \setlength\leftmarginii{\z@}%
  \addtolength\leftmarginii{0.5\labelsep}%
  \addtolength\leftmarginii{\@ACM@labelwidth}%
  \setlength\leftmarginiii{\leftmarginii}%
  \setlength\leftmarginiv{\leftmarginiii}%
  \setlength\leftmarginv{\leftmarginiv}%
  \setlength\leftmarginvi{\leftmarginv}%
  \@listi}
\makeatother
\usepackage{pifont,multirow,stfloats,makecell}
\usepackage{multicol}
\usepackage{circledsteps}
\usepackage{amsmath}
\usepackage{amsthm}
\usepackage{mathpartir}
\usepackage{mathtools}
\usepackage[inline]{enumitem}
\usepackage{graphicx}
\usepackage[singlelinecheck=false]{caption}
\usepackage{subcaption}
\usepackage{wrapfig}
\usepackage{tikz}
\usetikzlibrary{
  shapes,
  arrows,
  positioning,
  fit,
  backgrounds,
  calc,
  decorations.pathreplacing,
  chains,
  scopes,
  arrows.meta
}
\usepackage{pdflscape}
\usepackage{afterpage}
\usepackage{fancyvrb}

\definecolor{ocaml}{HTML}{bfffbf}
\definecolor{rust}{HTML}{ffbfbf}
\definecolor{debian}{HTML}{bfffff}
\definecolor{alpine}{HTML}{bfbfff}
\definecolor{linux}{HTML}{cccccc}

\newcommand\bcircle[2][]{\ifmmode
    \Circled[fill color=white,inner color=black,#1]{\mathsf{#2}}
  \else
    \Circled[fill color=white,inner color=black,#1]{\sffamily#2}
  \fi
}

\newsavebox{\sfboxa}\newsavebox{\sfboxb}\newsavebox{\sfboxc}\newsavebox{\sfboxd}
\newlength{\sfrowht}
\newcommand{\sfmax}[1]{\ifdim\dimexpr\ht#1+\dp#1\relax>\sfrowht \setlength{\sfrowht}{\dimexpr\ht#1+\dp#1\relax}\fi}
\newcommand{\sfcell}[1]{\begin{minipage}[c][\sfrowht][c]{\linewidth}\centering#1\end{minipage}}
\newlength{\sfrowhtb}
\newcommand{\sfmaxb}[1]{\ifdim\dimexpr\ht#1+\dp#1\relax>\sfrowhtb \setlength{\sfrowhtb}{\dimexpr\ht#1+\dp#1\relax}\fi}
\newcommand{\sfcellb}[1]{\begin{minipage}[c][\sfrowhtb][c]{\linewidth}\centering#1\end{minipage}}

\newtheoremstyle{theorem}
  {4pt}         %
  {4pt}         %
  {\itshape}    %
  {\parindent}  %
  {\scshape}    %
  {.}           %
  { }           %
  {}            %

\newtheoremstyle{definition}
  {4pt}         %
  {4pt}         %
  {}            %
  {\parindent}  %
  {\scshape}    %
  {.}           %
  { }           %
  {}            %

\theoremstyle{definition}
\newtheorem{theorem}{Theorem}[section]
\newtheorem{definition}[theorem]{Definition}

\makeatletter
\@addtoreset{theorem}{subsection}
\renewcommand{\thetheorem}{%
  \ifnum\c@subsection>0
    \thesubsection.\arabic{theorem}%
  \else
    \thesection.\arabic{theorem}%
  \fi
}
\makeatother

\newlist{subdefinition}{enumerate}{4}
\setlist[subdefinition, 1]{
  label=(\alph*),
  ref=\thedefinition.\alph*,
}
\setlist[subdefinition, 2]{
  label=(\roman*),
  ref=\thedefinition.\alph{subdefinitioni}.\roman*,
}
\setlist[subdefinition, 3]{
  label=(\Alph*),
  ref=\thedefinition.\alph{subdefinitioni}.\roman{subdefinitionii}.\Alph*,
}
\setlist[subdefinition, 4]{
  label=(\arabic*),
  ref=\thedefinition.\alph{subdefinitioni}.\roman{subdefinitionii}.\Alph{subdefinitioniii}.\arabic*,
}

\setcopyright{cc}
\setcctype{by}
\begin{document}
\title{Package Managers \`a la Carte}
\subtitle{A Formal Model of Dependency Resolution}

\author{Ryan T. Gibb}
\orcid{0009-0009-5702-3143}
\affiliation{%
  \institution{University of Cambridge}
  \city{Cambridge}
  \country{United Kingdom}
}
\email{rtg24@cam.ac.uk}

\author{Patrick Ferris}
\orcid{0000-0002-0778-8828}
\affiliation{%
  \institution{University of Cambridge}
  \city{Cambridge}
  \country{United Kingdom}
}
\email{pf341@cam.ac.uk}

\author{David Allsopp}
\orcid{0009-0006-2998-8210}
\affiliation{%
  \institution{Jane Street}
  \city{London}
  \country{United Kingdom}
}
\email{dallsopp@janestreet.com}

\author{Thomas Gazagnaire}
\orcid{0009-0001-6893-4630}
\affiliation{%
  \institution{Tarides}
  \city{Paris}
  \country{France}
}
\email{thomas@tarides.com}

\author{Anil Madhavapeddy}
\orcid{0000-0001-8954-2428}
\affiliation{%
  \institution{University of Cambridge}
  \city{Cambridge}
  \country{United Kingdom}
}
\email{avsm2@cam.ac.uk}

\renewcommand{\shortauthors}{Ryan T. Gibb, Patrick Ferris, David Allsopp, Thomas Gazagnaire, and Anil Madhavapeddy}

\begin{abstract}
  Package managers are legion.
  Every programming language and operating system has its own solution, each with subtly different semantics for dependency resolution.
  This fragmentation prevents multilingual projects from expressing precise dependencies across language ecosystems; it leaves external system dependencies implicit and unversioned; and it obscures the full dependency graph that supply-chain analysis depends on.
  We present the \textit{Package Calculus}, a formalism for dependency resolution that unifies the core semantics of package managers.
  Through a series of formal reductions, we show how this core is expressive enough to model the diversity of real-world dependency expression languages.
  The calculus provides the theoretical foundation for future cross-ecosystem tooling, as a \textit{lingua franca} of dependency expression.
\end{abstract}

\maketitle

\input{body}

\end{document}

%% file: body.tex
\vfill

\section{Introduction}
\label{sec.introduction}

Package managers are numerous and diverse; despite all fundamentally solving the same problem, they remain ad hoc and non-interoperable.
Consider a machine-learning project on Debian Linux, using a combination of C, Rust, and OCaml code for a static binary, with dynamic Python bindings that use GPU drivers, which in turn depend on a particular kernel driver.
Four separate package managers are required to deploy this program: opam~(OCaml), Cargo~(Rust), pip~(Python), and APT~(Debian).
If the program is to be portable beyond Debian, even more are required: APK~(Alpine) and DNF~(Red Hat-based Linux).
Dependencies between these ecosystems are managed ad hoc: they are typically unversioned, duplicated across ecosystems, and consequently invisible to security tooling.
Yet all these package managers share a common core: from dependencies written in a domain-specific language (DSL)~\cite{bentley1986dsl}, they \textit{resolve} a set of package versions that transitively satisfies them -- a \textit{resolution} -- and deploy it as a filesystem subtree.

\pagebreak

In this paper, we define a surprisingly small core calculus as a \textit{lingua franca} of dependency expression, and extensions for ecosystem-specific functionality that reduce back to the core.
This makes cross-ecosystem comparison precise,
so that superficially different mechanisms
can be identified as the same construct, or distinguished where they differ.
Each reduction from an extension to the core is a mechanism for emulating the extension in an ecosystem that lacks it natively, formalising workarounds practitioners already apply by hand.
We provide the foundation that principled cross-ecosystem tooling -- such as polyglot resolvers and translators between ecosystems -- has lacked to date.
We mechanise in Lean~4 the theorems presented in the paper~\cite{artifact}.

We make these specific contributions:
\begin{itemize}
  \item A survey of over thirty package managers across operating system and programming language ecosystems, identifying a shared core of packages, dependency relations, dependency resolution, and deployment; and axes along which they diverge~(\S\ref{sec.background}).
  \item The \textit{Package Calculus}, a minimal formal system which captures dependency resolution in three conditions -- root inclusion, dependency closure, and version uniqueness.
        We show how the NP-completeness of dependency resolution can be avoided, but is ultimately inherent in capturing the complexity of real-world package managers~(\S\ref{sec.calculus}).
  \item A model for each axis of divergence in dependency expression as an extension to the core calculus, with a sound and complete reduction from each extension back to the core~(\S\ref{sec.mise-en-place}).
        We are primarily concerned with satisfiability: the reductions preserve which resolutions are valid, not which are optimal.
  \item A demonstration that we can compose extensions to model package managers combining several at once, with general considerations for doing so.
        This opens the path to a \textit{polyglot resolver}: each ecosystem compiles its dependencies into the core and is returned a resolution to deploy -- so a project spanning several ecosystems is resolved as one.
        We further show how the core could serve as an intermediate representation for cross-ecosystem translation~(\S\ref{sec.a-la-carte}).
\end{itemize}

Package management to date has not unified the semantics of different systems; existing interoperability efforts have focused on sharing solving mechanisms between ecosystems with the same semantics~(\S\ref{sec.related}).
Instead of seeing package managers as monolithic and unrelated, we tease apart the commonality to show how we can unify the semantics of dependency expression~(\S\ref{sec.conclusion}).

\section{The Package Management Problem}
\label{sec.background}

Package managers all face the same problem: given a project's dependencies, choose a consistent set of package versions to install -- a resolution -- and deploy it~(\S\ref{sec.background.common-core}).
They underwent a Cambrian explosion in the mid-2000s~(Appendix~\ref{appendix.history}), and now diverge in the semantics of dependency expression~(\S\ref{sec.background.survey}).
Fig.~\ref{fig.glossary} collects the terminology used throughout the paper.

\begin{figure}[ht]
  \captionsetup{justification=centering}
  \centering
  \begin{tikzpicture}
    \node[draw=black!70, rounded corners=6pt, line width=0.8pt, inner sep=10pt] (glossarybox) {%
      \begin{minipage}{\dimexpr\linewidth-24pt\relax}
        \footnotesize
        \raggedright
        \setlength{\parskip}{2.5pt}%
        \setlength{\columnsep}{10pt}%
        \begin{multicols}{2}
          \textbf{Package}: a unit of software -- code, binaries, or data -- identified by a \textit{name} and \textit{version}\dotfill\S\ref{sec.background.common-core.package}\par
          \textbf{Dependency}: a relation from a \textit{depender} package, requiring one of a set of \textit{dependee} packages\dotfill\S\ref{sec.background.common-core.dependency}\par
          \textbf{Resolution}: the set of packages that satisfies a package's dependencies transitively, or the act of selecting such a set, computed by a \textit{resolver}\dotfill\S\ref{sec.background.common-core.resolution}\par
          \textbf{Deployment}: the provision of a resolution, typically as a filesystem subtree of installed packages\dotfill\S\ref{sec.background.common-core.deployment}\par
          \textbf{Ecosystem}: the programming language or operating system a package manager serves\dotfill\S\ref{sec.background.survey.ecosystem}\par
          \textbf{Toolchain integration}: the coupling of a package manager with build systems and compilers\dotfill\S\ref{sec.background.survey.toolchain}\par
          \textbf{Packaging language}: the DSL in which package metadata and dependencies are expressed\dotfill\S\ref{sec.background.survey.language}\par
          \textbf{Version constraint}: a dependency on a set of versions of a package name, e.g.\ $\geq\!1 \land \leq\!3$\dotfill\S\ref{sec.background.survey.version-constraints}\par
          \textbf{Conflict}: an `anti-dependency', preventing co-installation, e.g.\ of rival forks\dotfill\S\ref{sec.background.survey.conflicts}\par
          \textbf{Concurrent versions}: multiple versions of the same package name coexisting in a resolution\dotfill\S\ref{sec.background.survey.concurrent}\par
          \textbf{Peer dependency}: a child's constraint on a sibling's version, both dependees of one parent\dotfill\S\ref{sec.background.survey.peer}\par
          \textbf{Feature}: optional functionality of a package that may add dependencies, unified across dependers\dotfill\S\ref{sec.background.survey.features}\par
          \textbf{Package formula}: a Boolean (\texttt{and}/\texttt{or}/\texttt{not}) expression over package dependencies\dotfill\S\ref{sec.background.survey.package-formula}\par
          \textbf{Variable formula}: a package formula parameterised by variables, e.g.\ the operating system\dotfill\S\ref{sec.background.survey.variable-formula}\par
          \textbf{Virtual package}: a package name that multiple packages can provide, e.g.\ an SSH server\dotfill\S\ref{sec.background.survey.virtual-packages}\par
          \textbf{Optional dependency}: a use-if-present dependency where the dependee is never required, but used for build ordering if it is in the resolution\dotfill\S\ref{sec.background.survey.optional-dependencies}\par
        \end{multicols}
      \end{minipage}};
    \node[fill=white, inner xsep=6pt, inner ysep=2pt] at (glossarybox.north) {\small\textsc{Glossary}};
  \end{tikzpicture}
  \Description{A boxed glossary of the package management terms used in the paper.}
  \caption{A glossary of package management terminology as used in this paper.}
  \label{fig.glossary}
\end{figure}

\subsection{A Common Core}
\label{sec.background.common-core}

We identify concerns common to package managers.
The Package Calculus formalises packages, dependencies, and resolutions~(\S\ref{sec.calculus}); deployment affects the semantics of dependency expression~(\S\ref{sec.background.survey}).

\subsubsection{Package}
\label{sec.background.common-core.package}
A \textit{package} is a unit of software -- source code, a compiled binary, or data -- distributed with metadata that declares the other packages it depends on.
Packages are identified by a \textit{name}, which refers to a package across its releases, and a \textit{version}, distinguishing one release from another.

\subsubsection{Dependency}
\label{sec.background.common-core.dependency}
Package managers express \textit{dependency} relations from a \textit{depender} package to a set of \textit{dependee} packages, meaning the depender requires one of the dependees.
Dependees are often software libraries, while packages that are not dependees are typically applications intended for use by an end user.
Dependees are typically expressed as a name and set of compatible versions.

\begin{figure}[ht]
  \captionsetup{justification=centering}
  \centering
  \includegraphics[scale=0.8]{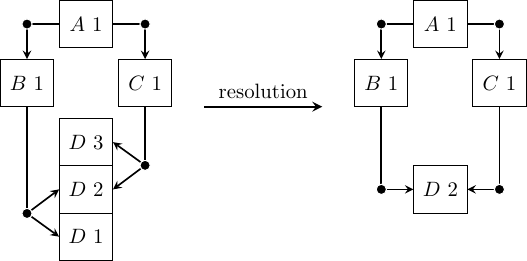}
  \Description{
    Two node-and-edge diagrams joined by a rightward arrow labelled `resolution'.
    On the left, the dependency structure: a root package A version 1 depends on B version 1 and C version 1; B depends on version 1 or 2 of package name D, and C on version 2 or 3, so all three versions of D appear as candidates.
    On the right, the resulting resolution: the same packages with only D version 2 retained -- the single version satisfying both B and C -- and versions 1 and 3 removed.
  }
  \caption{Resolving a set of dependencies (left) with root package name $A$ and version $1$ yields a \textit{resolution} -- the set of packages that satisfies them -- by selecting a single version of each package name (right).}
  \label{fig.resolution}
\end{figure}

\subsubsection{Resolution}
\label{sec.background.common-core.resolution}
Dependency \textit{resolution} is the problem of computing, given a package, the set of package names and versions needed to satisfy its dependencies transitively.
A \textit{resolver} performs dependency resolution.
A \textit{resolution} is the set of packages selected by \textit{resolving} dependencies.
Fig.~\ref{fig.resolution} illustrates a resolution that satisfies three conditions: the root package is included (root inclusion); every package's dependencies are met by a package in the resolution (dependency closure); and at most one version of each name is chosen (version uniqueness).
We formalise these conditions in~\S\ref{sec.calculus}.

\subsubsection{Deployment}
\label{sec.background.common-core.deployment}
Given a \textit{resolution}, package managers provide packages through \textit{deployment}.
Interpreted-language package managers can place source code in a directory, like npm~(JavaScript).
Others build packages from metadata instructions, as Cargo, opam, or Cabal~(Haskell) do.
System package managers might unpack binary archives into the filesystem hierarchy.
For embedded deployments, as with BusyBox or Yocto, the result may be a whole firmware image.
Deployment is orthogonal to resolution, but its constraints determine the semantics of dependency expression.
For example, requiring unique paths or linker symbols per package name forces version uniqueness.

\input{table}

\subsection{A Survey of Package Managers}
\label{sec.background.survey}

Table~\ref{tbl.comparison} compares a selection of package managers, chosen to cover the most widely used systems as well as those with distinctive semantics, and assembled from reading official documentation, inspecting source code, and consulting with contributors.
We define the table's columns below.

\subsubsection{Ecosystems}
\label{sec.background.survey.ecosystem}
Package managers fall into two broad categories: language package managers distribute libraries and tools to developers in a specific programming language; system package managers administer a full operating system.
The lines between these often blur; system package managers do distribute language libraries, and language package managers can facilitate setting up a development environment with system dependencies.
Almost all package managers group packages in a centralised repository for discoverability and curation, which ranges from validating version constraints~(\S\ref{sec.background.survey.version-constraints}) to maintaining a coherent package set, a single mutually compatible version per name replacing version constraints, like Homebrew~(macOS), Nix, \TeX\ Live, and Stackage~(Haskell).

\subsubsection{Toolchain Integration}
\label{sec.background.survey.toolchain}
Some package managers, especially language package managers, are integrated with other parts of the toolchain.
opam includes build scripts in package metadata which can invoke a build system or compiler, but is language-agnostic.
Other ecosystems have tighter integration; Cargo is Rust's package manager \textit{and} build system.
This enables functionality impossible without a well-defined API between these parts of the toolchain.
For example, Cargo mangles symbols with package versions to support concurrent versions~(\S\ref{sec.background.survey.concurrent}).

\subsubsection{Packaging Language}
\label{sec.background.survey.language}
Packaging DSLs are normally static data structures for expressing dependencies~(\S\ref{sec.background.common-core.dependency}) -- e.g.\ TOML, \texttt{.deb} control files, or \texttt{.opam} files.
Sometimes DSLs are expressions embedded in Turing-complete languages~\cite{hudak1996dsel} -- e.g.\ Homebrew's Ruby, Spack's~(high-performance computing) Python, or Portage's~(Gentoo Linux) ebuild scripts -- that compile to a static structure.
We next cover the varied functionality these DSLs support in expressing dependencies.

\subsubsection{Version Constraints}
\label{sec.background.survey.version-constraints}
A \textit{version formula} such as $\geq 1 \land \leq 3$ enumerates to a set $\{1, 2, 3\}$.
Some package managers, such as Nix~(Various), can only express a dependency on a single version of a package name.
Go's Minimal Version Selection (MVS) specifies only a minimum version bound for each dependency, and has a linear-time dependency resolution algorithm~(\S\ref{sec.calculus.resolution-complexity}).

\subsubsection{Conflicts}
\label{sec.background.survey.conflicts}
A conflict is an `anti-dependency' on the absence of a package name and version set, used, for example, for forks of a project that are not co-installable.
Portage distinguishes \textit{weak blocks} (\texttt{!pkg}) from \textit{strong blocks} (\texttt{!!pkg}): weak blocks are non-binding preferences for the resolver, analogous to Debian's \texttt{Recommends}~(\S\ref{sec.background.survey.optional-dependencies}), while strong blocks enforce hard mutual exclusion.

\subsubsection{Concurrent Versions}
\label{sec.background.survey.concurrent}
Many package managers allow only one version of a package name in a resolution due to deployment constraints such as requiring unique symbols when linking objects, or storing packages at a unique path in the filesystem.
Package managers can use isolated environments to work around this, such as opam switches or Python virtual environments.

In contrast, package managers with concurrent versions allow multiple versions of a name to exist within a resolution.
Cargo uses name mangling to link multiple versions of a library into a binary.
Cargo restricts this to packages with different major versions under the semantic versioning scheme: co-linking minor versions would create duplicate, incompatible types in the binary, while different major versions are expected to have distinct APIs~\cite{cargoresolver, semver}.
Concurrent versions are therefore isolated -- the types of one major version cannot interoperate with another.
Nix~\cite{dolstra2004nix} diverges from the Linux Filesystem Hierarchy Standard to place packages at paths containing unique cryptographic hashes, allowing the installation of multiple versions of a package name on the same system.
npm allows duplication of dependee packages by storing them in a subdirectory of their depender.

\subsubsection{Peer Dependencies}
\label{sec.background.survey.peer}
Given a depender parent and a dependee child, a peer dependency is a constraint that the child places on the version of its `peer' sibling: another dependee of the parent.
This is typically used by plugins; for example, a React component library (the child) declares a peer dependency on \texttt{react} (the sibling).
It does not provide its own copy, but instead requires that its parent provide a compatible version.
This functionality is only required with concurrent versions~(\S\ref{sec.background.survey.concurrent}), as otherwise the child could simply depend on the peer directly.

{\emergencystretch=1em%
npm supports peer dependencies, but their behaviour has changed multiple times.
The \texttt{-{}-legacy-peer-deps} behaviour is that a peer dependee is only included when the parent also depends on it.
The contemporary behaviour is that peer dependees that are not marked as optional are included regardless of the parent's dependencies.\par}

Peer dependencies appear only in npm's row of Table~\ref{tbl.comparison}; npm stands in for a family of JavaScript package managers (pnpm, Bun, Yarn, etc.), omitted as they share its semantics.
The mechanism itself is general -- any concurrent-version ecosystem (e.g.\ Cargo) could adopt it -- while single-version ecosystems enforce this implicitly, since every depender sees the same version.

\subsubsection{Features}
\label{sec.background.survey.features}
{\emergencystretch=1em%
In Cargo, dependency relations are parameterised by features -- optional functionality of a package that can introduce additional dependencies.
For example,
\begin{center}
  \texttt{image = \{ version = "0.25", features = ["png"] \}}
\end{center}
selects the \texttt{image} crate's \texttt{png} feature, causing \texttt{image} to additionally depend on the \texttt{png} decoder crate.
The resolver unifies features, so a package will be provided with the union of all features requested by its dependers.
In the Python world, dependency specifiers define `extras' which can similarly add dependencies.
Portage's \textit{USE flags} serve the same role: a dependency can request USE flags with \texttt{dev-libs/foo[bar]}, and USE flags can trigger additional dependencies via \mbox{\texttt{bar?\ (\ dev-libs/baz\ )}}.
USE flags additionally allow a package to depend on the \textit{absence} of a flag -- \texttt{dev-libs/foo[-bar]} in Portage -- combining features with conflicts~(\S\ref{sec.background.survey.conflicts}).
Cargo features are only additive: a depender can require a feature but not block its presence.\par}

{\emergencystretch=1em%
Ecosystems without native feature support sometimes encode features manually.
opam's \texttt{ocaml-variants} package encodes compiler options in version strings such as \texttt{4.02.1+fp}: each combination is a separate version, of which version uniqueness~(\S\ref{sec.calculus.core}) permits one; without unification, every union of options is its own version, so their number grows combinatorially.
Newer \texttt{ocaml-option-*} packages instead encode options as distinct package names, e.g.\ \texttt{ocaml-option-flambda} -- structurally similar to the feature reduction of Def.~\ref{def.feature-reduction} -- with the \texttt{ocaml-options-only-*} family enforcing exclusive combinations via explicit conflicts~(\S\ref{sec.background.survey.conflicts}).\par}

\subsubsection{Package Formulae}
\label{sec.background.survey.package-formula}
opam allows a package to declare that it requires either \texttt{lwt} or \texttt{async} as a concurrency monad, e.g.\ \texttt{"lwt" \{>= "6.0"\}\,|\,"async"}.
Similarly, APT supports disjunction in \texttt{Depends} fields, for example, \texttt{Depends: libssl-dev\,|\,libgnutls-dev} to accept either TLS implementation.

\subsubsection{Variable Formulae}
\label{sec.background.survey.variable-formula}
{\emergencystretch=1em%
Variable formulae extend package formulae~(\S\ref{sec.background.survey.package-formula}).
For example, in opam, the variable \texttt{os-distribution} makes a formula conditional on the operating system on which the package manager is deploying the package.
In APT, dependencies can similarly be parameterised by computer architecture.
The opam package-local \texttt{with-test} denotes dependencies that are only required when tests have been enabled for the package.
Python's dependency groups serve a similar role, having groups such as \texttt{docs} or \texttt{test} in \texttt{pyproject.toml} labelling sets of dependencies that are conditional on the development task.
Variables can also be used to distinguish build-time from run-time dependencies.
Instead of an explicit Boolean formula, these can also be implicit labels on dependencies, making them conditional on the platform or architecture.\par}

\subsubsection{Virtual Packages}
\label{sec.background.survey.virtual-packages}
Virtual packages let a depender express a need for some capability without naming a specific provider; capable packages declare themselves providers of the shared virtual name.
For example, in APT~\cite{debian}, both \texttt{openssh-server} and \texttt{dropbear-bin} provide \texttt{ssh-server}.

\subsubsection{Optional Dependencies}
\label{sec.background.survey.optional-dependencies}
opam supports optional dependencies, where dependees are used by the depender if present in a resolution, determining build ordering and triggering rebuilds.
Like features~(\S\ref{sec.background.survey.features}), they enable functionality, but their effect is resolution-wide rather than scoped to individual dependencies.
While such use-if-present behaviour occurs throughout toolchains -- with Gentoo going so far as to outlaw unrecorded \textit{automagic dependencies}, and sandboxed builds used to avoid them -- only opam in Table~\ref{tbl.comparison} makes it explicit.
Optional dependencies are a purely post-resolution concern, so we omit them from the model (Appendix~\ref{appendix.optional-dependencies}).

opam's optional dependencies are distinct from npm's \texttt{optionalDependencies} (like Debian's \texttt{Recommends}), which ask the resolver to include a dependee if possible but do not fail otherwise.
An extension with resolver preferences could capture this best-effort behaviour.

\section{The Package Calculus}
\label{sec.calculus}

We define a minimal formal system for dependency resolution in package management systems~(\S\ref{sec.calculus.core}), extend it with version formulae~(\S\ref{sec.calculus.version-formulae}), and show how the NP-completeness of resolution can be avoided, but is ultimately inherent in capturing the complexity of real-world package managers~(\S\ref{sec.calculus.resolution-complexity}).

\subsection{The Core Calculus}
\label{sec.calculus.core}

We define the core calculus, the minimal kernel of the Package Calculus.

\begin{definition}[Packages]
  \label{def.calculus.package}
  We define:
  \begin{subdefinition}
    \item\label{def.calculus.package.names}
    $N$ as the set of possible package names.
    \item\label{def.calculus.package.versions}
    $V$ as the set of possible package versions.
    \item\label{def.calculus.package.packages}
    $N\times V$ as the set of possible packages, each a name-version pair.
    \item\label{def.calculus.package.real}
    $R \subseteq N \times V$ as the set of \textit{real} packages -- those that exist.
  \end{subdefinition}
\end{definition}

\noindent Throughout, any variable left unbound is implicitly universally quantified.

\begin{definition}[Dependency]
  \label{def.calculus.dependency}
  We define dependencies $\Delta \subseteq (N \times V) \times (N \times \mathcal{P}(V))$ as a relation from packages to a name and set of versions.
  As is standard for a relation, $p \Delta (n, vs)$ abbreviates $(p, (n, vs)) \in \Delta$, where a depender $p \in N \times V$ expresses a dependency on dependees of name $n \in N$ with versions $vs \subseteq V$.
  We assume that dependencies are \textit{functional in name}, where $p \Delta (n, vs) \land p \Delta (n, vs') \implies vs = vs'$.
  Dependencies that are not functional in name can be normalised by merging same-name entries per depender, intersecting their version sets.
\end{definition}

\begin{definition}[Resolution]
  \label{def.calculus.resolution}
  Given dependencies $\Delta$ and root $r \in R$, a resolution $S \subseteq R$ is valid if the following conditions hold:
  \begin{subdefinition}
    \item\label{def.calculus.resolution.root-inclusion}
    \textbf{Root Inclusion}: $r \in S$
    \item\label{def.calculus.resolution.dependency-closure}
    \textbf{Dependency Closure}: $\forall\, p \in S.\, p \Delta (n, vs) \implies \exists\, v \in vs.\, (n, v) \in S$ \\
    Every dependency of a package in the resolution must be satisfied by a compatible version.
    \item \label{def.calculus.resolution.version-uniqueness}
    \textbf{Version Uniqueness}: $\forall\, (n, v),\, (n, v') \in S.\, v = v'$
  \end{subdefinition}
  An installation request is a \textit{query} $Q \subseteq N \times \mathcal{P}(V)$: a set of package names, each with a set of acceptable versions.
  As the formalism takes only a single root package as input, we realise a query as a synthetic root $r \in R$ whose dependencies are exactly $Q$, that is, $r \Delta (m, vs) \iff (m, vs) \in Q$.

  We write $\mathcal{S}(\Delta, r)$ for the set of all resolutions of $r$ in $\Delta$.
  Read $S \in \mathcal{S}(\Delta, r)$ as `$S$ is a valid resolution'.
  \textsc{DependencyResolution} is the decision problem: given $R$, $\Delta$, and $r$, is $\mathcal{S}(\Delta, r)$ non-empty?
\end{definition}

\begin{figure}[ht]
  \captionsetup{justification=centering}
  \begin{lrbox}{\sfboxa}\begin{minipage}{0.27\textwidth}%
\begin{Verbatim}[fontsize=\footnotesize]
Package: A
Depends: B (= 1), C (= 1)

Package: B
Depends: D (>= 1), D (<< 3)

Package: C
Depends: D (>= 2)
\end{Verbatim}
  \end{minipage}\end{lrbox}%
  \begin{lrbox}{\sfboxb}\begin{minipage}{0.27\textwidth}\centering
    \[\begin{aligned}
      (A, 1) &\Delta (B, \{1\}) \\
      (A, 1) &\Delta (C, \{1\}) \\
      (B, 1) &\Delta (D, \{1, 2\}) \\
      (C, 1) &\Delta (D, \{2, 3\})
    \end{aligned}\]
  \end{minipage}\end{lrbox}%
  \begin{lrbox}{\sfboxc}\begin{minipage}{0.27\textwidth}\centering
    \includegraphics[scale=0.8]{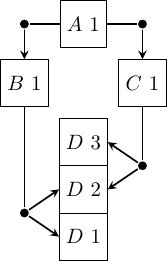}%
  \end{minipage}\end{lrbox}%
  \setlength{\sfrowht}{0pt}\sfmax{\sfboxa}\sfmax{\sfboxb}\sfmax{\sfboxc}%
  \hfill
  \begin{subfigure}[b]{0.27\textwidth}
    \sfcell{\usebox{\sfboxa}}
    \caption{Debian syntax.}
    \label{fig.calculus.syntax}
  \end{subfigure}%
  \hfill
  \begin{subfigure}[b]{0.27\textwidth}
    \sfcell{\usebox{\sfboxb}}
    \caption{Dependency relation $\Delta$.}
    \label{fig.calculus.dependencies}
  \end{subfigure}%
  \hfill
  \begin{subfigure}[b]{0.27\textwidth}
    \sfcell{\usebox{\sfboxc}}
    \Description{Hypergraph rendering of the dependency relation $\Delta$ given alongside.}
    \caption{Hypergraph illustration.}
    \label{fig.calculus.illustration}
  \end{subfigure}%
  \hfill\null
  \caption{An instance of the Package Calculus.}
  \label{fig.calculus}
\end{figure}

We illustrate these definitions with an example.
Fig.~\ref{fig.calculus.syntax} shows a snippet of Debian syntax for expressing the dependencies in Fig.~\ref{fig.calculus.dependencies}.
We visualise these dependencies in Fig.~\ref{fig.calculus.illustration} as a directed hypergraph -- a generalisation of a directed graph in which each edge connects a set of source vertices to a set of target vertices~\cite{berge1970hypergraphs}.
Each dependency is an edge from the depender (the domain) to a set of dependees.
Formally, we define vertices $R$ and edges
\[
  E = \{(\{(n, v)\}, \{(m, u) \mid u \in vs\}) \mid (n, v) \Delta (m, vs)\}
\]
Note that we restrict the domain to a cardinality of one -- we can only express a dependency \textit{from} one package.
The only valid resolution\footnote{An alternative core allows multiple versions and encodes uniqueness via conflicts, but many package managers do not support conflicts (Table~\ref{tbl.comparison}); we model the majority case directly, and concurrent-version systems extend it~(\S\ref{sec.mise-en-place.concurrent}).} for query $Q = \{(A, \{1\})\}$ is
\[
  S = \{r, (A, 1), (B, 1), (C, 1), (D, 2)\}
\]

\begin{theorem}
  \label{thm.resolution-complexity}
  \textsc{DependencyResolution} (Def.~\ref{def.calculus.resolution}) is NP-complete.
\end{theorem}
\begin{proof}
  By a reduction from \textsc{3-SAT}: a variable becomes a package name with versions $\top$ and $\bot$, and a clause becomes a package name, required by the root, with one version per literal -- each depending on the assignment that satisfies that literal.
  See Appendix~\ref{appendix.resolution-complexity}.
\end{proof}

\subsection{Version Formulae}
\label{sec.calculus.version-formulae}

The core calculus represents each dependency as an explicit set of compatible versions (Def.~\ref{def.calculus.dependency}).
Real package managers, however, do not enumerate these sets: they describe them with \textit{version formulae}~(\S\ref{sec.background.survey.version-constraints}) interpreted over a total ordering of versions that is typically non-lexicographical.
We show that the core's abstract treatment of versions loses nothing by working with sets: a version formula reduces to the core by evaluating it to its set of admitted versions.
This is also the simplest instance of the extend-and-reduce pattern of~\S\ref{sec.mise-en-place} -- extend the calculus, define resolution over the extension, then reduce back to the core.

\begin{definition}[Version Ordering]
  \label{def.version-ordering}
  A \textit{version ordering} is a total order $\leq_v$ on $V$ (Def.~\ref{def.calculus.package.versions}).
\end{definition}

Different ecosystems define different version orderings and formulae~\cite{semver, debianversions}.
The VERS specification~\cite{vers} defines a universal version range notation, implemented by the univers library~\cite{univers}, which wraps ecosystem-specific version comparison logic.
We define the Version Formula Package Calculus to be as expressive as the VERS specification according to the following definitions.

\pagebreak

\begin{definition}[Version Formula]
  \label{def.version-formula}
  We define:
  \begin{subdefinition}
    \item\label{def.version-formula.formulae}
    Version formulae $\Phi$:
    \[\phi ::= \top \mid \bot \mid \phi \land \phi \mid \phi \lor \phi \mid \omega\ c \qquad \omega ::=\ \geq\ \mid\ >\ \mid\ \leq\ \mid\ <\ \mid\ =\ \mid\ \neq \qquad c \in V\]
    \item\label{def.version-formula.real}
    $R_\Phi$, defined precisely as $R$ in Def.~\ref{def.calculus.package.real}, the subscript serving merely to differentiate it from the $R$ produced by the reduction to the core.
    \item\label{def.version-formula.versions}
    $V_n = \{ v \mid (n, v) \in R_\Phi \}$, the set of versions of $n \in N$ that exist.
    \item\label{def.version-formula.ordering}
    A version ordering $\leq_v$ (Def.~\ref{def.version-ordering}).
    \item\label{def.version-formula.semantics}
    Semantics function $\llbracket \cdot \rrbracket_n : \Phi \to \mathcal{P}(V)$ for $\Phi$ under $n \in N$, where $\omega_v$ interprets $\omega$ under $\leq_v$:
    \begin{gather*}
      \begin{alignedat}{2}
        \llbracket \top \rrbracket_n &= V_n &\qquad
        \llbracket \phi_1 \land \phi_2 \rrbracket_n &= \llbracket \phi_1 \rrbracket_n \cap \llbracket \phi_2 \rrbracket_n \\
        \llbracket \bot \rrbracket_n &= \emptyset &\qquad
        \llbracket \phi_1 \lor \phi_2 \rrbracket_n &= \llbracket \phi_1 \rrbracket_n \cup \llbracket \phi_2 \rrbracket_n
      \end{alignedat}
      \qquad
      \llbracket \omega\, c \rrbracket_n = \{v \in V_n \mid v\ \omega_v\ c\}
    \end{gather*}
  \end{subdefinition}
\end{definition}

These version formulae can express common version constraints:
\begin{gather*}
  \begin{alignedat}{2}
    \texttt{">1.0.0"} &\mapsto \llbracket > 1.0.0 \rrbracket_n &\qquad
    \texttt{"!=1.4.*"} &\mapsto \llbracket < 1.4.0\ \lor\ \geq 1.5.0 \rrbracket_n \\
    \texttt{"{\textasciicircum}1.2.3"} &\mapsto \llbracket \geq 1.2.3\ \land\ < 2.0.0 \rrbracket_n&\qquad
    \texttt{"{\textasciitilde}1.2.3"} &\mapsto \llbracket \geq 1.2.3\ \land\ < 1.3.0 \rrbracket_n
  \end{alignedat}
\end{gather*}

We extend the calculus to use version formulae as the representation of versions in dependencies (Def.~\ref{def.version-formula-dependency}), define what a valid resolution looks like over this representation (Def.~\ref{def.version-formula-resolution}), and reduce the extension back to the core calculus by evaluating each formula to its version set (Def.~\ref{def.version-formula-reduction}).

\begin{definition}[Version Formula Dependency]
  \label{def.version-formula-dependency}
  We define version formula dependencies $\Delta_\Phi \subseteq (N \times V) \times (N \times \Phi)$ where an element $p \Delta_\Phi (n, \phi)$ has formula $\phi$ representing the set of compatible versions of $n$.
\end{definition}

\begin{definition}[Version Formula Resolution]
  \label{def.version-formula-resolution}
  Given dependencies $\Delta_\Phi$ and root $r_\Phi \in R_\Phi$, a resolution $S_\Phi \subseteq R_\Phi$ is valid if:
  \begin{subdefinition}
    \item\label{def.version-formula-resolution.root-inclusion}
    \textbf{Root Inclusion}: $r_\Phi \in S_\Phi$, as in Def.~\ref{def.calculus.resolution.root-inclusion}.
    \item\label{def.version-formula-resolution.dependency-closure}
    \textbf{Dependency Closure}: $\forall\, p \in S_\Phi.\, p \Delta_\Phi (n, \phi) \implies \exists\, v \in \llbracket \phi \rrbracket_n.\, (n, v) \in S_\Phi$
    \item\label{def.version-formula-resolution.version-uniqueness}
    \textbf{Version Uniqueness}: $\forall\, (n, v),\, (n, v') \in S_\Phi.\, v = v'$, as in Def.~\ref{def.calculus.resolution.version-uniqueness}.
  \end{subdefinition}
  We write $\mathcal{S}_\Phi(\Delta_\Phi, r_\Phi)$ for the set of all resolutions of $r_\Phi$ in $\Delta_\Phi$.
\end{definition}

\begin{definition}[Version Formula Reduction]
  \label{def.version-formula-reduction}
  Given $R_\Phi$, $\Delta_\Phi$, and $r_\Phi$, we define a reduction to the core calculus with $R$, $\Delta$, and $r$ as follows:
  \begin{subdefinition}
    \item\label{def.version-formula-reduction.packages}
    \textbf{Packages}: $r = r_\Phi$ and $R = R_\Phi$
    \item\label{def.version-formula-reduction.dependencies}
    \textbf{Dependencies}:
    $\Delta = \{ ((n, v), (m, \llbracket \phi \rrbracket_m)) \mid (n, v) \Delta_\Phi (m, \phi) \}$
  \end{subdefinition}
\end{definition}

\begin{theorem}[Correctness]
  \label{thm.version-formula-reduction-correctness}
  $S \in \mathcal{S}(\Delta, r) \iff S \in \mathcal{S}_\Phi(\Delta_\Phi, r_\Phi)$
\end{theorem}

\subsection{Resolution Complexity}
\label{sec.calculus.resolution-complexity}

Package managers navigate the NP-completeness of \textsc{DependencyResolution} (Thm.~\ref{thm.resolution-complexity}) in different ways, forming a spectrum from restricted-but-tractable resolution to full NP-completeness.
At the most restricted end, Go's MVS~(\S\ref{sec.background.survey.version-constraints})~\cite{cox2018mvs} achieves deterministic, linear-time resolution.
MVS relies on strict adherence to semantic versioning~\cite{semver}:
only major version changes can break APIs.
Under this assumption, dependencies specify only a minimum version bound, since any version above the minimum within the same major version is guaranteed to be compatible.
Different major versions are treated as different package names, allowing multiple major versions to coexist in the same resolution.
Rather than selecting the latest available version, MVS selects the \textit{minimum} version satisfying all bounds -- the maximum of all minimum bounds.
Since what constitutes a breaking change is ambiguous in practice, selecting the minimum reduces exposure to unintended breakage.
Upgrades are an explicit developer action that rewrites the minimum bounds.

Cargo also allows multiple versions of a package name to exist in a resolution, but it falls short of linear-time resolution because its dependency formulae allow upper bounds within a major version -- minor version bumps sometimes include breaking changes.
Its dependency formulae can also span semver-incompatible~\cite{semver} version sets (e.g.\ \texttt{">=1, <3"}, \texttt{">=0.1, <0.5"}), whereas one can only depend on \textit{one} major version of a Go module.

At the other end, in ecosystems like OCaml, any API change is potentially breaking due to features such as module inclusion, where a package's API is implicitly dependent on the APIs of its dependees.
There is no meaningful distinction between major and minor versions, making semantic versioning inapplicable and MVS unviable.
Were opam to adopt MVS, it would be unable to express compatibility with multiple major versions, and so could depend on only single, exact versions.
For this reason, opam specifies no upper bounds by default; instead, when a new release is proposed, the opam-repository continuous integration (CI) builds and tests all packages that depend on it (its reverse dependencies), and failures result in a retroactive upper bound on the new release~\cite{ocamllabs-year}.

Semantic versioning is a useful heuristic but not a substitute for expressive version constraints.
Regardless of ecosystem, versioning is determined by the dependee, but incompatibility is a property of the dependency relation; one dependency's compatible change is another's breaking change.

Nix sidesteps the spectrum by treating every name-version pair as a distinct package and requiring the exact version of a dependee to be specified -- singular dependencies~(Appendix~\ref{appendix.singular-dependencies}).
Derivations are parameterised by their dependencies in the Nix DSL, so resolution is performed manually rather than by a constraint solver.
In the Nixpkgs package repository, typically one version of a package name is included at a time, and breakages are often due to incompatibilities introduced by updates.

We can formalise the approaches underlying this spectrum~\cite{cox2016versionsat} as two independent ways to reduce the complexity of resolution:
\begin{enumerate}
  \item\label{sec.calculus.resolution-complexity.total-ordering}
        \textbf{Restrict version constraints to lower bounds.}
        With a total ordering $\leq_v$ of versions (Def.~\ref{def.version-ordering}), we redefine a dependency relation to only have a minimum bound, with a dependency as $p \Delta_{\min} (n, v)$ where $p \in N \times V$, $n \in N$, $v \in V$, and $v$ is the minimum bound.
        We replace the dependency closure condition (Def.~\ref{def.calculus.resolution.dependency-closure}) with
        \[\forall\, p \in S.\, p \Delta_{\min} (n, v) \implies \exists\, v' \in V.\, (n, v') \in S \land v \leq_v v'\]
        An incompatibility between a depender and a dependee satisfying this relationship is considered a specification error; the depender and dependee authors must collaborate to fix it.
        Without upper bounds, any version above the maximum of all minimum bounds is valid, so a greedy algorithm suffices -- selecting the minimum satisfying version for stability, as MVS does, or the latest for freshness.
        This reduces the problem to a graph traversal and version selection in $O(|\Delta_{\min}| + |R|)$ time.
  \item\label{sec.calculus.resolution-complexity.version-unique}
        \textbf{Remove version uniqueness.}
        Remove the version uniqueness condition (Def.~\ref{def.calculus.resolution.version-uniqueness}), allowing multiple versions of a package name.
        Each dependency can be resolved independently, allowing a greedy algorithm to select any compatible version for each dependency.
        This reduces the problem to graph traversal in $O(|\Delta|)$ time.
\end{enumerate}

Go's MVS combines both approaches -- the version uniqueness constraint is removed for different major versions~(approach~\ref{sec.calculus.resolution-complexity.version-unique}), and within a major version, version constraints are restricted to a lower bound~(approach~\ref{sec.calculus.resolution-complexity.total-ordering}).
Cargo also removes the version uniqueness constraint for different major versions, but the practical need for upper bounds within a major version reintroduces NP-completeness.
opam exploits neither, accepting NP-complete resolution in exchange for expressiveness, at the cost of extensive compatibility testing.
Nix adopts approach~\ref{sec.calculus.resolution-complexity.version-unique} as well as a restriction stronger than approach~\ref{sec.calculus.resolution-complexity.total-ordering}: singular dependencies, where each dependency specifies exactly one version, delegating version selection from the constraint solver to a Turing-complete packaging language.

{\emergencystretch=1em%
Where a package manager sits on this spectrum determines the sophistication of its resolver: greedy algorithms suffice at the restricted end, while those accepting full NP-completeness employ SAT solvers (Appendix~\ref{appendix.sat-resolution}) or implement conflict-driven clause learning (CDCL) natively for performance and error reporting.
Similarly, lock files~\cite{gamage2026} -- snapshots of a resolution -- are unnecessary under MVS, where the resolution is determined by the minimum bounds alone, but are necessary for reproducibility where the resolution changes as new versions are released.\par}

Our theory was informed by the goal of practical resolution across ecosystems~\cite{decan2018,dietrich2019}.
To this end, singular dependencies and Go's MVS are not expressive enough to model general dependency resolution, while the Turing-complete Nix DSL is too expressive to model in other dependency resolvers and does not (yet) support resolving version constraints.
We argue that the Package Calculus is the \textit{minimum} degree of expressiveness needed to model the ecosystems surveyed.

\section{Package Managers, Mise en Place}
\label{sec.mise-en-place}

We model each axis of divergence in dependency expression~(\S\ref{sec.background}) by extending the core calculus~(\S\ref{sec.calculus.core}) and reducing each extension back to the core.
We define the Conflict Package Calculus~(\S\ref{sec.mise-en-place.conflicts}), Concurrent Package Calculus~(\S\ref{sec.mise-en-place.concurrent}), Peer Package Calculus~(\S\ref{sec.mise-en-place.peer-dependencies}), Feature Package Calculus~(\S\ref{sec.mise-en-place.features}), Package Formula Calculus~(\S\ref{sec.mise-en-place.package-formula}), Variable Formula Calculus~(\S\ref{sec.mise-en-place.variable-formula}), and Virtual Package Calculus~(\S\ref{sec.mise-en-place.virtual}).
We cover singular dependencies, which are less expressive than the core, in Appendix~\ref{appendix.singular-dependencies}; and optional dependencies, which have no resolution-time effect, in Appendix~\ref{appendix.optional-dependencies}.
All reductions are polynomial in the size of the input instance.
The reductions serve as a compilation pass from the extended calculus to the core.
The soundness theorems of the reductions construct a resolution in the extended calculus from a resolution of the reduction.
We prepare each extension in isolation -- a \textit{mise en place} -- before composing them to order in~\S\ref{sec.a-la-carte}.

\subsection{Conflicts}
\label{sec.mise-en-place.conflicts}

The core calculus has no mechanism for one package to forbid another~(\S\ref{sec.background.survey.conflicts}).
We extend it with a conflict relation (Def.~\ref{def.conflict}), define what a valid resolution looks like over the extended calculus (Def.~\ref{def.conflict-resolution}), and reduce back to the core by encoding each conflict as a synthetic package (Def.~\ref{def.conflict-reduction}).

\begin{definition}[Conflict]
  \label{def.conflict}
  We define conflicts $\Gamma \subseteq (N \times V) \times (N \times \mathcal{P}(V))$ as a relation where an element $p \Gamma (n, vs)$ denotes that package $p \in N \times V$ conflicts with package name $n \in N$ with versions $vs \subseteq V$.
\end{definition}

We use conflicts $\Gamma$ as the defining symbol of the calculus and introduce $\Delta_\Gamma$ as $\Delta$ in Def.~\ref{def.calculus.dependency}, and $R_\Gamma$ as $R$ in Def.~\ref{def.calculus.package.real}, in order to differentiate them from the use of $\Delta$ and $R$ in the reduction to the core.
In subsequent extensions, unless we explicitly redefine a subscript of a core calculus definition -- as we do with version formula dependency $\Delta_\Phi$ in Def.~\ref{def.version-formula-dependency} -- read it as defined in the core.

\begin{definition}[Conflict Resolution]
  \label{def.conflict-resolution}
  Given dependencies $\Delta_\Gamma$, conflicts $\Gamma$, and root $r_\Gamma$, a resolution $S_\Gamma \subseteq R_\Gamma$ is valid if:
  \begin{subdefinition}
    \item\label{def.conflict-resolution.core-resolution}
    \textbf{Core Resolution}: $S_\Gamma \in \mathcal{S}(\Delta_\Gamma, r_\Gamma)$ (Def.~\ref{def.calculus.resolution}).
    \item\label{def.conflict-resolution.conflict-avoidance}
    \textbf{Conflict Avoidance}: $\forall\, p \in S_\Gamma.\, p \Gamma (n, vs) \implies \not\exists\, v \in vs.\, (n, v) \in S_\Gamma$ \\
    For every package in the resolution, a conflicting package must not be in the resolution.
  \end{subdefinition}
  We write $\mathcal{S}_\Gamma(\Delta_\Gamma, \Gamma, r_\Gamma)$ for the set of all resolutions of $r_\Gamma$ in $\Delta_\Gamma$ and $\Gamma$.
\end{definition}

We reduce the Conflict Package Calculus to the core calculus by encoding each conflict as a synthetic conflict package name with two versions, $1$ and $0$.
The package declaring the conflict depends on version $1$, while each conflicting package depends on version $0$.
Since both versions cannot coexist by version uniqueness, the two sides of the conflict are mutually exclusive.
Fig.~\ref{fig.conflict} shows an instance of the Conflict Package Calculus reduced to the core.

\begin{figure}[ht]
  \captionsetup{justification=centering}
  \begin{lrbox}{\sfboxa}\begin{minipage}{0.20\textwidth}%
\begin{Verbatim}[fontsize=\footnotesize]
Package: A
Version: 1
Conflicts: B (<< 3)
\end{Verbatim}
  \end{minipage}\end{lrbox}%
  \begin{lrbox}{\sfboxb}\begin{minipage}{0.17\textwidth}\centering
    \includegraphics[scale=0.8]{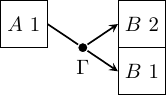}%
  \end{minipage}\end{lrbox}%
  \begin{lrbox}{\sfboxc}\begin{minipage}{0.32\textwidth}\centering
    \includegraphics[scale=0.8]{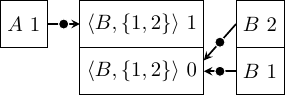}%
  \end{minipage}\end{lrbox}%
  \begin{lrbox}{\sfboxd}\begin{minipage}{0.25\textwidth}\centering
    \[\begin{aligned}
      (A, 1) &\Delta (\langle B, \{1, 2\} \rangle, \{1\}) \\
      (B, 1) &\Delta (\langle B, \{1, 2\} \rangle, \{0\}) \\
      (B, 2) &\Delta (\langle B, \{1, 2\} \rangle, \{0\})
    \end{aligned}\]
  \end{minipage}\end{lrbox}%
  \setlength{\sfrowht}{0pt}\sfmax{\sfboxa}\sfmax{\sfboxb}\sfmax{\sfboxc}\sfmax{\sfboxd}%
  \hfill
  \begin{subfigure}[b]{0.20\textwidth}
    \sfcell{\usebox{\sfboxa}}
    \caption{Debian syntax.}
    \label{fig.conflict.debian}
  \end{subfigure}%
  \hfill
  \begin{subfigure}[b]{0.18\textwidth}
    \sfcell{\usebox{\sfboxb}}
    \Description{Hypergraph rendering of the conflict relation given alongside.}
    \caption{Hypergraph.}
    \label{fig.conflict.hypergraph}
  \end{subfigure}%
  \hfill
  \begin{subfigure}[b]{0.32\textwidth}
    \sfcell{\usebox{\sfboxc}}
    \Description{Hypergraph rendering of the reduction formula given alongside.}
    \caption{Hypergraph of reduction.}
    \label{fig.conflict.hypergraph-reduction}
  \end{subfigure}%
  \hfill
  \begin{subfigure}[b]{0.25\textwidth}
    \sfcell{\usebox{\sfboxd}}
    \caption{Reduction.}
    \label{fig.conflict.reduction}
  \end{subfigure}%
  \hfill\null
  \caption{Conflict Package Calculus $(A, 1) \Gamma (B, \{1, 2\})$ reduced to the core calculus.}
  \label{fig.conflict}
\end{figure}

\begin{definition}[Conflict Reduction]
  \label{def.conflict-reduction}
  Given $R_\Gamma$, $\Delta_\Gamma$, $\Gamma$, and $r_\Gamma$, we define a reduction to the core calculus with $R$, $\Delta$, and $r$ as follows:
  \begin{subdefinition}
    \item\label{def.conflict-reduction.packages}
    \textbf{Packages}:
      $r = r_\Gamma$ and, with $\langle n, vs \rangle \in N$ and $0, 1 \in V$,
      \[R = R_\Gamma \cup \bigcup\limits_{p \Gamma (n, vs)} \{(\langle n, vs \rangle, 0), (\langle n, vs \rangle, 1)\}\]
    \item\label{def.conflict-reduction.dependencies}
    \textbf{Dependencies}:
    \[
      \Delta = \Delta_\Gamma
      \cup \bigcup_{p \Gamma (n, vs)} \{(p,\, (\langle n, vs \rangle,\, \{1\}))\}
      \cup \bigcup_{p \Gamma (n, vs)} \{((n, u),\, (\langle n, vs \rangle,\, \{0\})) \mid u \in vs\}
    \]
  \end{subdefinition}
\end{definition}

We define soundness and completeness theorems for each reduction.
Given an instance of the extended calculus, soundness gives a construction that extracts a resolution of the extended instance from a resolution of its core reduction.
Completeness goes the other way: a construction of a resolution of the core reduction from a resolution of the extended instance.

\begin{theorem}[Soundness]
  \label{thm.conflict-reduction-soundness}
  If $S \in \mathcal{S}(\Delta, r)$, then
  $S_\Gamma = S \cap R_\Gamma$
  is a valid resolution in $\mathcal{S}_\Gamma(\Delta_\Gamma, \Gamma, r_\Gamma)$.
\end{theorem}

\begin{theorem}[Completeness]
  \label{thm.conflict-reduction-completeness}
  If $S_\Gamma \in \mathcal{S}_\Gamma(\Delta_\Gamma, \Gamma, r_\Gamma)$, then
  \[
    S = S_\Gamma
    \cup \bigcup_{p \Gamma (n, vs)} \{ (\langle n, vs \rangle, 1) \mid p \in S_\Gamma \}
    \cup \bigcup_{p \Gamma (n, vs)} \{ (\langle n, vs \rangle, 0) \mid u \in vs,\ (n, u) \in S_\Gamma \}
  \]
  is a valid resolution in $\mathcal{S}(\Delta, r)$.
\end{theorem}

This mechanism generalises to a mutual-exclusion \textit{conflict class} by giving the synthetic name one version per member, encoding the $n$-way exclusion in $O(n)$ instead of $O(n^2)$ pairwise conflicts.

\subsection{Concurrent Versions}
\label{sec.mise-en-place.concurrent}

Fig.~\ref{fig.diamond} illustrates an example of the `diamond dependency problem'; $\{(A, B), (A, C), (B, D), (C, D)\}$ as edges form a diamond.
We define the Concurrent Package Calculus to model package managers that address this by allowing multiple versions of a package name to exist in the resolution~(\S\ref{sec.background.survey.concurrent}).

\begin{figure}[ht]
  \captionsetup{justification=centering}
  \begin{lrbox}{\sfboxa}\begin{minipage}{0.20\textwidth}%
{\footnotesize\bfseries\ttfamily B/Cargo.toml}
\begin{Verbatim}[fontsize=\footnotesize]
[dependencies]
D = "=1"
\end{Verbatim}
{\footnotesize\bfseries\ttfamily C/Cargo.toml}
\begin{Verbatim}[fontsize=\footnotesize]
[dependencies]
D = "=3"
\end{Verbatim}
  \end{minipage}\end{lrbox}%
  \begin{lrbox}{\sfboxb}\begin{minipage}{0.20\textwidth}\centering
    \[\begin{aligned}
      (A, 1) &\Delta (B, \{1\}) \\
      (A, 1) &\Delta (C, \{1\}) \\
      (B, 1) &\Delta (D, \{1\}) \\
      (C, 1) &\Delta (D, \{3\})
    \end{aligned}\]
  \end{minipage}\end{lrbox}%
  \begin{lrbox}{\sfboxc}\begin{minipage}{0.25\textwidth}\centering
    \includegraphics[scale=0.8]{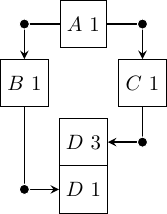}%
  \end{minipage}\end{lrbox}%
  \setlength{\sfrowht}{0pt}\sfmax{\sfboxa}\sfmax{\sfboxb}\sfmax{\sfboxc}%
  \hfill
  \begin{subfigure}[b]{0.20\textwidth}
    \sfcell{\usebox{\sfboxa}}
    \caption{Cargo syntax.}
    \label{fig.diamond.cargo}
  \end{subfigure}%
  \hfill
  \begin{subfigure}[b]{0.20\textwidth}
    \sfcell{\usebox{\sfboxb}}
    \caption{Package Calculus.}
    \label{fig.diamond.calculus}
  \end{subfigure}%
  \hfill
  \begin{subfigure}[b]{0.25\textwidth}
    \captionsetup{justification=centering}
    \sfcell{\usebox{\sfboxc}}
    \Description{Hypergraph rendering of the diamond dependency given alongside.}
    \caption{Hypergraph illustration.}
    \label{fig.diamond.hypergraph}
  \end{subfigure}%
  \hfill\null
  \caption{The `diamond dependency problem', which with $Q = \{(A, \{1\})\}$ has no valid resolution.}
  \label{fig.diamond}
\end{figure}

\begin{definition}[Granularity Function]
  \label{def.granularity-function}
  We define a set of granularities $G$ and a function $g : V \to G$ mapping each version to a granularity such that versions that cannot coexist share a granularity.
  For example, in Cargo, we can define $g(x.y.z) = x$ to extract the major version of a package, and for npm or Nix $g(v)=v$ as they allow duplication regardless of version.
  We can emulate the core with $g(v)=\epsilon$.
\end{definition}

\begin{definition}[Concurrent Resolution]
  \label{def.concurrent-resolution}
  Given dependencies $\Delta_C$, granularity $g$, and root $r_C$, a resolution $S_C \subseteq R_C$ and parent relation $\pi \subseteq S_C \times S_C$ are valid if:
  \begin{subdefinition}
    \item\label{def.concurrent-resolution.root-inclusion}
    \textbf{Root Inclusion}: $r_C \in S_C$, as in Def.~\ref{def.calculus.resolution.root-inclusion}.
    \item\label{def.concurrent-resolution.parent-closure}
    \textbf{Parent Closure}:
    $\forall\, p \in S_C.\, p \Delta_C (n, vs) \implies \exists!\, v \in vs.\, (n, v) \in S_C \land ((n, v), p) \in \pi$ \\
    A parent of a package is the depender in the dependency that caused the child's inclusion, and each parent has exactly one child per dependency.
    \item\label{def.concurrent-resolution.version-granularity}
    \textbf{Version Granularity}:
    $\forall\, (n, v),\, (n, v') \in S_C.\, v\neq v' \implies g(v)\neq g(v')$ \\
    Different versions of a name in $S_C$ must differ under $g$.
  \end{subdefinition}
  We write $\mathcal{S}_C(\Delta_C, g, r_C) \ni (S_C, \pi)$ for all resolution-parent pairs of $r_C$ in $\Delta_C$ under $g$.
\end{definition}

To reduce the Concurrent Package Calculus to the core calculus, we push the granular version into the package name so that version uniqueness in the core calculus enforces version granularity.
We introduce intermediate packages that allow for depending on multiple granular versions of now-distinct package names~\cite{pubgrub-rs}.
The parent relation $\pi$ and the intermediate packages of the reduction below determine one another (Thms~\ref{thm.concurrent-reduction-soundness} and~\ref{thm.concurrent-reduction-completeness}); peer dependencies~(\S\ref{sec.mise-en-place.peer-dependencies}) then constrain precisely this selection.
 npm represents $\pi$ in its lock file as nested \verb|node_modules| paths.
Fig.~\ref{fig.concurrent} shows an instance of the Concurrent Package Calculus and its reduction to the core.

\begin{figure}[ht]
  \captionsetup{justification=centering}
  \begin{lrbox}{\sfboxa}\begin{minipage}{0.16\textwidth}%
{\footnotesize\bfseries\ttfamily B/Cargo.toml}
\begin{Verbatim}[fontsize=\footnotesize]
[dependencies]
D = ">=1, <3"
\end{Verbatim}
{\footnotesize\bfseries\ttfamily C/Cargo.toml}
\begin{Verbatim}[fontsize=\footnotesize]
[dependencies]
D = ">=2, <4"
\end{Verbatim}
  \end{minipage}\end{lrbox}%
  \begin{lrbox}{\sfboxb}\begin{minipage}{0.36\textwidth}\centering
    \[\begin{aligned}
      (A, 1.0.0) &\Delta_C (B, \{1.0.0\}) \\
      (A, 1.0.0) &\Delta_C (C, \{1.0.0\}) \\
      (B, 1.0.0) &\Delta_C (D, \{1.0.0, 2.0.0, 2.0.1\}) \\
      (C, 1.0.0) &\Delta_C (D, \{2.0.0, 2.0.1, 3.0.0\})
    \end{aligned}\]
  \end{minipage}\end{lrbox}%
  \begin{lrbox}{\sfboxc}\begin{minipage}{0.31\textwidth}\centering
    \includegraphics[scale=0.8]{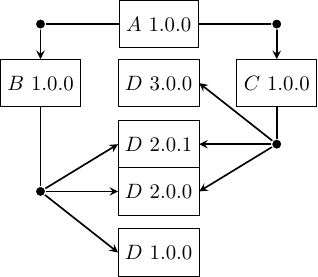}%
  \end{minipage}\end{lrbox}%
  \setlength{\sfrowht}{0pt}\sfmax{\sfboxa}\sfmax{\sfboxb}\sfmax{\sfboxc}%
  \hfill
  \begin{subfigure}[b]{0.16\textwidth}
    \sfcell{\usebox{\sfboxa}}
    \caption{Cargo syntax.}
    \label{fig.concurrent.cargo}
  \end{subfigure}%
  \hfill
  \begin{subfigure}[b]{0.36\textwidth}
    \sfcell{\usebox{\sfboxb}}
    \caption{Concurrent Package Calculus.}
    \label{fig.concurrent.calculus}
  \end{subfigure}%
  \hfill
  \begin{subfigure}[b]{0.31\textwidth}
    \sfcell{\usebox{\sfboxc}}
    \Description{Hypergraph rendering of the concurrent calculus dependencies given alongside.}
    \caption{Hypergraph illustration.}
    \label{fig.concurrent.hypergraph}
  \end{subfigure}%
  \hfill\null

  \vspace{1em}
  \begin{lrbox}{\sfboxa}\begin{minipage}{0.40\textwidth}\centering
    \[\begin{aligned}
      (\langle A, 1 \rangle, 1.0.0) &\Delta (\langle B, 1 \rangle, \{1.0.0\}) \\
      (\langle A, 1 \rangle, 1.0.0) &\Delta (\langle C, 1 \rangle, \{1.0.0\}) \\
      (\langle B, 1 \rangle, 1.0.0) &\Delta (\langle B, 1.0.0, D \rangle, \{1, 2\}) \\
      (\langle B, 1.0.0, D \rangle, 1) &\Delta (\langle D, 1 \rangle, \{1.0.0\}) \\
      (\langle B, 1.0.0, D \rangle, 2) &\Delta (\langle D, 2 \rangle, \{2.0.0, 2.0.1\}) \\
      (\langle C, 1 \rangle, 1.0.0) &\Delta (\langle C, 1.0.0, D \rangle, \{2, 3\}) \\
      (\langle C, 1.0.0, D \rangle, 2) &\Delta (\langle D, 2 \rangle, \{2.0.0, 2.0.1\}) \\
      (\langle C, 1.0.0, D \rangle, 3) &\Delta (\langle D, 3 \rangle, \{3.0.0\})
    \end{aligned}\]
  \end{minipage}\end{lrbox}%
  \begin{lrbox}{\sfboxb}\begin{minipage}{0.58\textwidth}\centering
    \includegraphics[scale=0.8]{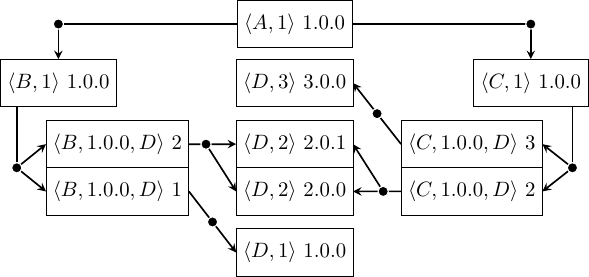}%
  \end{minipage}\end{lrbox}%
  \setlength{\sfrowht}{0pt}\sfmax{\sfboxa}\sfmax{\sfboxb}%
  \begin{subfigure}[b]{0.40\textwidth}
    \sfcell{\usebox{\sfboxa}}
    \caption{Reduction to the core calculus.}
    \label{fig.concurrent-reduction.reduction}
  \end{subfigure}%
  \hfill
  \begin{subfigure}[b]{0.58\textwidth}
    \sfcell{\usebox{\sfboxb}}
    \Description{Hypergraph rendering of the reduction formula given alongside.}
    \caption{Hypergraph of reduction.}
    \label{fig.concurrent-reduction.hypergraph}
  \end{subfigure}%
  \hfill\null
  \caption{Concurrent Package Calculus instance where $g(x.y.z) = x$, reduced to the core calculus.}
  \label{fig.concurrent}
\end{figure}

\begin{definition}[Concurrent Reduction]
  \label{def.concurrent-reduction}
  Given $R_C$, $\Delta_C$, $g$, and $r_C$, we reduce to the core with $R$, $\Delta$, and $r$.
  For each $(n, v) \Delta_C (m, vs)$ let $W = \{ g(u) \mid u \in vs \}$, with $\langle n, g(v) \rangle, \langle n, v, m \rangle \in N$ and $W \subseteq V$.
  \begin{subdefinition}
    \item\label{def.concurrent-reduction.packages}
    \textbf{Packages}: $r = (\langle n, g(v) \rangle, v)$ where $r_C = (n, v)$, and
    \[
      R = \bigcup_{(n, v) \in R_C} \{ (\langle n, g(v) \rangle, v) \}
      \cup \bigcup_{(n, v) \Delta_C (m, vs)} \{ (\langle n, v, m \rangle, w) \mid \lvert W \rvert > 1,\ w \in W \}
    \]
    \item\label{def.concurrent-reduction.dependencies}
    \textbf{Dependencies}: each dependency reduces to a dependency on single-granularity dependees, or to an intermediate $\langle n, v, m \rangle$ with each version depending on a granularity group:
    \[
      \Delta = \bigcup_{(n, v) \Delta_C (m, vs)}
      \begin{cases}
        \{ ((\langle n, g(v) \rangle, v),\, (\langle m, w \rangle,\, vs)) \} & \text{if } W = \{w\} \\[4pt]
        \left(\begin{aligned}
          &\{ ((\langle n, g(v) \rangle, v),\, (\langle n, v, m \rangle,\, W)) \} \\
          {}\cup{}&\{ ((\langle n, v, m \rangle, w),\, (\langle m, w \rangle,\, \{ u \in vs \mid g(u) = w \})) \mid w \in W \}
        \end{aligned}\right) & \text{if } \lvert W \rvert > 1 \\[8pt]
        \{ ((\langle n, g(v) \rangle, v),\, (\langle n, v, m \rangle,\, \emptyset)) \} & \text{if } W = \emptyset
      \end{cases}
    \]
  \end{subdefinition}
\end{definition}

\begin{theorem}[Soundness]
  \label{thm.concurrent-reduction-soundness}
  If $S \in \mathcal{S}(\Delta, r)$, then $S_C = \{(n, v) \mid (\langle n, g(v) \rangle, v) \in S\}$ and
  \[
    \pi = \bigcup_{(n, v) \Delta_C (m, vs)} \left\{ ((m, u),\, (n, v)) \;\middle|\;
    \begin{aligned}
      &u \in vs,\ (n, v) \in S_C,\ (m, u) \in S_C, \\
      &\lvert W \rvert > 1 \implies (\langle n, v, m \rangle, g(u)) \in S
    \end{aligned}
    \right\}
  \]
  are a valid resolution and parent relation in $\mathcal{S}_C(\Delta_C, g, r_C)$.
\end{theorem}

\begin{theorem}[Completeness]
  \label{thm.concurrent-reduction-completeness}
  If $(S_C, \pi) \in \mathcal{S}_C(\Delta_C, g, r_C)$, then
  \[
    \begin{aligned}
      S = {}& \{ (\langle n, g(v) \rangle, v) \mid (n, v) \in S_C \} \\
      {}\cup{}& \bigcup_{(n, v) \Delta_C (m, vs)} \left\{ (\langle n, v, m \rangle, g(u)) \;\middle|\;
      \begin{aligned}
        &(n, v) \in S_C,\ (m, u) \in S_C,\ u \in vs, \\
        &((m, u), (n, v)) \in \pi,\ \lvert W \rvert > 1
      \end{aligned}
      \right\}
    \end{aligned}
  \]
  is a valid resolution in $\mathcal{S}(\Delta, r)$.
\end{theorem}

The Concurrent Package Calculus allows multiple versions of the same package name to be selected even if it would be possible to use just one, consistent with Cargo's implementation.
Future work could restrict versions for \textit{public} dependencies exposed in a depender's interface, which must share one version across their subgraph -- a compatibility question that type interfaces make precise~\cite{florisson2016packages} -- as in Cargo's proposed public/private dependencies~\cite{cargo-pub-priv-deps}.

\subsection{Peer Dependencies}
\label{sec.mise-en-place.peer-dependencies}

A peer dependency~(\S\ref{sec.background.survey.peer}) is a child's constraint on a peer the parent must also depend on -- meaningful only when multiple versions can coexist, so we define the Peer Package Calculus atop the Concurrent Package Calculus~(\S\ref{sec.mise-en-place.concurrent}).
Our formalisation captures npm's \texttt{-{}-legacy-peer-deps} behaviour, though the contemporary behaviour can also be modelled with minor modifications.

\begin{definition}[Peer Dependency]
  \label{def.peer-dependency}
  We define peer dependencies $\Theta \subseteq (N \times V) \times (N \times \mathcal{P}(V))$ as a relation where an element $p \Theta (n, vs)$ denotes that a parent of package $p$ can only depend on the peer package name $n$ with version in $vs$.
\end{definition}

\begin{definition}[Peer Dependency Resolution]
  \label{def.peer-dependency-resolution}
  Given dependencies $\Delta_C$, peer dependencies $\Theta$, granularity $g$, and root $r_C$, a resolution $S_\Theta \subseteq R_C$ and parent relation $\pi \subseteq S_\Theta \times S_\Theta$ are valid if:
  \begin{subdefinition}
    \item\label{def.peer-dependency-resolution.concurrent-resolution}
    \textbf{Concurrent Resolution}: $(S_\Theta, \pi)$ is valid in $\mathcal{S}_C(\Delta_C, g, r_C)$ (Def.~\ref{def.concurrent-resolution}).
    \item\label{def.peer-dependency-resolution.peer-satisfaction}
    \textbf{Peer Satisfaction}:
    The version of $n$ that $p$'s parent $q$ selected must satisfy $p$'s peer constraint.
    \[\forall\, p \in S_\Theta.\, p \Theta (n, vs) \land (p, q) \in \pi \land q \Delta_C (n, us) \land v \in us \land ((n, v), q) \in \pi \implies v \in vs\]
  \end{subdefinition}
  We write $\mathcal{S}_\Theta(\Delta_C, \Theta, g, r_C) \ni (S_\Theta, \pi)$ for all resolution-parent pairs of $r_C$ in $\Delta_C$ and $\Theta$ under $g$.
\end{definition}

\begin{figure}[ht]
  \captionsetup{justification=centering}
  \begin{lrbox}{\sfboxa}\begin{minipage}{0.19\textwidth}%
{\footnotesize\bfseries\ttfamily A/package.json}
\begin{Verbatim}[fontsize=\footnotesize]
"dependencies":
{
  "B": "1",
  "C": ">=2 <4"
}
\end{Verbatim}
{\footnotesize\bfseries\ttfamily B/package.json}
\begin{Verbatim}[fontsize=\footnotesize]
"peerDependencies":
{
  "C": ">=1 <3"
}
\end{Verbatim}
  \end{minipage}\end{lrbox}%
  \begin{lrbox}{\sfboxb}\begin{minipage}{0.21\textwidth}\centering
    \[\begin{aligned}
      (A, 1) &\Delta_C (B, \{1\}) \\
      (A, 1) &\Delta_C (C, \{2, 3\}) \\
      (B, 1) &\Theta (C, \{1, 2\})
    \end{aligned}\]
  \end{minipage}\end{lrbox}%
  \begin{lrbox}{\sfboxc}\begin{minipage}{0.18\textwidth}\centering
    \includegraphics[scale=0.8]{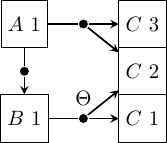}%
  \end{minipage}\end{lrbox}%
  \begin{lrbox}{\sfboxd}\begin{minipage}{0.42\textwidth}\centering
    \includegraphics[scale=0.8]{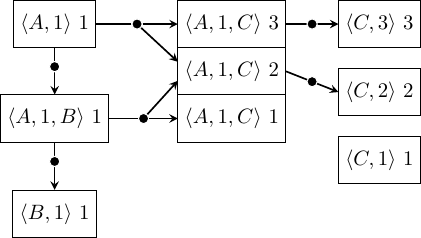}%
  \end{minipage}\end{lrbox}%
  \setlength{\sfrowht}{0pt}\sfmax{\sfboxb}\sfmax{\sfboxc}\sfmax{\sfboxd}\sfmax{\sfboxa}%
  \hfill
  \begin{subfigure}[b]{0.19\textwidth}
    \sfcell{\usebox{\sfboxa}}
    \caption{npm syntax.}
    \label{fig.peer-dependency.npm}
  \end{subfigure}%
  \hfill
  \begin{subfigure}[b]{0.21\textwidth}
    \sfcell{\usebox{\sfboxb}}
    \caption{Peer Package Calculus.}
    \label{fig.peer-dependency.calculus}
  \end{subfigure}%
  \hfill
  \begin{subfigure}[b]{0.18\textwidth}
    \sfcell{\usebox{\sfboxc}}
    \Description{Hypergraph rendering of the peer dependency relation given alongside.}
    \caption{Hypergraph.}
    \label{fig.peer-dependency.hypergraph}
  \end{subfigure}%
  \hfill
  \begin{subfigure}[b]{0.42\textwidth}
    \sfcell{\usebox{\sfboxd}}
    \Description{
      Hypergraph rendering of the Peer Dependency Reduction applied to the instance given alongside, with dependencies:
      (<A,1>, 1) on (<A,1,B>, \{1\});
      (<A,1,B>, 1) on (<B,1>, \{1\});
      (<A,1>, 1) on (<A,1,C>, \{2,3\});
      (<A,1,C>, 2) on (<C,2>, \{2\});
      (<A,1,C>, 3) on (<C,3>, \{3\});
      and, (<A,1,B>, 1) on (<A,1,C>, \{1,2\}).
    }
    \caption{Hypergraph of reduction.}
    \label{fig.peer-dependency.hypergraph-reduction}
  \end{subfigure}%
  \hfill\null
  \caption{Peer Package Calculus instance where $g(v)=v$, reduced to the core calculus.\\
  $(\langle A, 1, C\rangle, 1)$ can never be selected so it has no dependency on $(\langle C, 1\rangle, \{ 1 \})$.}
  \label{fig.peer-dependency}
\end{figure}

We reduce the Peer Package Calculus to the core calculus by modifying the Concurrent Reduction (Def.~\ref{def.concurrent-reduction}), using intermediate package dependencies to constrain the peer package version selection.
Fig.~\ref{fig.peer-dependency} shows an instance of the Peer Package Calculus reduced to the core.

\begin{definition}[Peer Dependency Reduction]
  \label{def.peer-dependency-reduction}
  Given $R_C$, $\Delta_C$, $\Theta$, $g$, and $r_C$, we define a reduction to the core calculus with $R$, $\Delta$, and $r$.
  For each dependency $(n, v) \Delta_C (m, vs)$ let $i = \langle n, v, m \rangle$, with $\langle n, g(v) \rangle, i \in N$, and let $\mathit{Peers} = \{ (o, u, ws) \mid (n, v) \Delta_C (o, us),\ u \in us,\ (o, u) \Theta (m, ws) \}$, the sibling instances constraining $m$.
  \begin{subdefinition}
    \item\label{def.peer-dependency-reduction.packages}
    \textbf{Packages}: $r = (\langle n, g(v) \rangle, v)$ where $r_C = (n, v)$, and
    \[
      R = \bigcup_{(n, v) \in R_C} \{ (\langle n, g(v) \rangle, v) \}
      \cup \bigcup_{(n, v) \Delta_C (m, vs)} \big(
        \{ (i, u) \mid u \in vs \}
        \cup \{ (i, w) \mid (o, u, ws) \in \mathit{Peers},\ w \in ws \}
      \big)
    \]
    \item\label{def.peer-dependency-reduction.dependencies}
    \textbf{Dependencies}: each dependency reduces to a dependency on an intermediate name $i$ with versions $vs$, each of which depends on the corresponding dependee.
    A peer dependency on $m$ from a sibling $o$ extends $i$ with the peer versions and makes $o$'s intermediate packages depend on $i$.
    We use the full version in the intermediate package to allow selection of an exact peer package version as the intersection of constraints, regardless of $g$.
    \[
      \Delta = \bigcup_{(n, v) \Delta_C (m, vs)} \left(
        \begin{aligned}
          &\{ ((\langle n, g(v) \rangle, v),\, (i,\, vs)) \} \\
          {}\cup{}&\{ ((i, u),\, (\langle m, g(u) \rangle,\, \{u\})) \mid u \in vs \} \\
          {}\cup{}&\{ ((\langle n, v, o \rangle, u),\, (i,\, ws)) \mid (o, u, ws) \in \mathit{Peers} \}
        \end{aligned}
      \right)
    \]
  \end{subdefinition}
\end{definition}

\begin{theorem}[Soundness]
  \label{thm.peer-dependency-reduction-soundness}
  If $S \in \mathcal{S}(\Delta, r)$, then $S_\Theta = \big\{(n, v) \mid (\langle n, g(v) \rangle, v) \in S\big\}$ and
  \[
    \pi = \bigcup_{(n, v) \Delta_C (m, vs)} \{ ((m, u),\, (n, v)) \mid u \in vs,\ (n, v) \in S_\Theta,\ (\langle n, v, m \rangle, u) \in S \}
  \]
  are a valid resolution and parent relation in $\mathcal{S}_\Theta(\Delta_C, \Theta, g, r_C)$.
\end{theorem}

\begin{theorem}[Completeness]
  \label{thm.peer-dependency-reduction-completeness}
  If $(S_\Theta, \pi) \in \mathcal{S}_\Theta(\Delta_C, \Theta, g, r_C)$, then
  \[
    S = \bigcup_{(n, v) \in S_\Theta} \{ (\langle n, g(v) \rangle, v) \}
    \cup \bigcup_{(n, v) \Delta_C (m, vs)} \left\{ (\langle n, v, m \rangle, u) \;\middle|\;
    \begin{aligned}
      &(n, v) \in S_\Theta,\ (m, u) \in S_\Theta, \\
      &u \in vs,\ ((m, u), (n, v)) \in \pi
    \end{aligned}
    \right\}
  \]
  is a valid resolution in $\mathcal{S}(\Delta, r)$.
\end{theorem}

\subsection{Features}
\label{sec.mise-en-place.features}

Features~(\S\ref{sec.background.survey.features}) are optional functionality of a package that may add dependencies, with selections unified across all dependers.
The core has no notion of parameterised dependency or unification; we extend it with both.

\begin{definition}[Feature]
  \label{def.feature}
  We define:
  \begin{subdefinition}
    \item\label{def.feature.features}
    $F$ as the finite set of possible features.
    \item\label{def.feature.support}
    $\mathit{support} \subseteq (N \times V) \times F$ as a relation where $(p, f) \in \mathit{support}$ denotes that package $p$ supports feature $f$.
  \end{subdefinition}
\end{definition}

\begin{definition}[Feature Dependency]
  \label{def.feature-dependency}
  We define:
  \begin{subdefinition}
    \item\label{def.feature-dependency.parameterised}
    Parameterised dependencies $\Delta_f \subseteq (N \times V) \times (N \times \mathcal{P}(V) \times \mathcal{P}(F))$ as a relation where an element $p \Delta_f (n, vs, fs)$ denotes that package $p \in N \times V$ depends on package name $n \in N$, with a set of compatible versions $vs \subseteq V$, and a set of required features $fs \subseteq F$.
    \item\label{def.feature-dependency.additional}
    Additional dependencies $\Delta_a \subseteq ((N \times V) \times F) \times (N \times \mathcal{P}(V) \times \mathcal{P}(F))$ as a relation where an element $(p, f) \Delta_a (n, vs, fs)$ denotes that package $p \in N \times V$ with feature $f \in F$ depends on $(n, vs, fs)$ as in $\Delta_f$.
  \end{subdefinition}
\end{definition}

\begin{definition}[Feature Resolution]
  \label{def.feature-resolution}
  Given parameterised dependencies $\Delta_f$, additional dependencies $\Delta_a$, and root $r_f$ with no features ($\forall\, f.\, (r_f, f) \notin \mathit{support}$), a resolution $S_f \subseteq R_f \times \mathcal{P}(F)$ -- where an element $(p, fs) \in S_f$ represents package $p$ selected with features $fs \subseteq F$ -- is valid if:
  \begin{subdefinition}
    \item\label{def.feature-resolution.root-inclusion}
    \textbf{Root Inclusion}: $(r_f, \emptyset) \in S_f$
    \item\label{def.feature-resolution.parameterised-dependency-closure}
    \textbf{Parameterised Dependency Closure}:
    \[\forall\, (p, fs_p) \in S_f.\, p \Delta_f (n, vs, fs) \implies \exists\, v \in vs,\, fs' \supseteq fs.\, ((n, v), fs') \in S_f\]
    A dependency parameterised by features must be satisfied by a package with a compatible version and selected features.
    \item\label{def.feature-resolution.additional-dependency-closure}
    \textbf{Additional Dependency Closure}:
    \[\forall\, (p, fs_p) \in S_f,\, f \in fs_p.\, (p, f) \Delta_a (n, vs, fs) \implies \exists\, v \in vs,\, fs' \supseteq fs.\, ((n, v), fs') \in S_f\]
    A dependency added by a selected feature of the depender must be satisfied as with parameterised dependencies.
    \item\label{def.feature-resolution.feature-unification}
    \textbf{Feature Unification}: $\forall\, ((n, v), fs),\, ((n, v'), fs') \in S_f.\, fs = fs'$
    \item\label{def.feature-resolution.version-uniqueness}
    \textbf{Version Uniqueness}: $\forall\, ((n, v), fs),\, ((n, v'), fs') \in S_f.\, v = v'$
    \item\label{def.feature-resolution.support}
    \textbf{Support}: $\forall\, (p, fs) \in S_f,\, f \in fs.\, (p, f) \in \mathit{support}$ \\
    Every selected feature of a package must be in the support relation.
  \end{subdefinition}
  We write $\mathcal{S}_f(\Delta_f, \Delta_a, r_f)$ for the set of all resolutions of $r_f$ in $\Delta_f$ and $\Delta_a$.
\end{definition}

To reduce the Feature Package Calculus to the core calculus, we create packages representing features selected for a package~\cite{pubgrub-rs}.
Fig.~\ref{fig.features} shows an instance of the Feature Package Calculus and its reduction to the core.

\begin{figure}[ht]
  \captionsetup{justification=centering}
  \begin{lrbox}{\sfboxa}\begin{minipage}{0.30\textwidth}\centering
    \begin{align*}
      (A, 1) &\Delta_f (B, \{1\}, \emptyset) \\
      (A, 1) &\Delta_f (C, \{1\}, \emptyset) \\
      (B, 1) &\Delta_f (D, \{1\}, \{\alpha, \beta\}) \\
      (C, 1) &\Delta_f (D, \{1\}, \{\beta\}) \\
      ((D, 1), \alpha) &\Delta_a (E, \{1\}, \emptyset) \\
      ((D, 1), \beta) &\Delta_a (F, \{1\}, \emptyset)
    \end{align*}
  \end{minipage}\end{lrbox}%
  \begin{lrbox}{\sfboxb}\begin{minipage}{0.30\textwidth}\centering
    \includegraphics[scale=0.8]{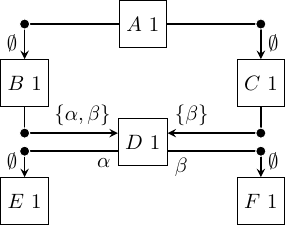}%
  \end{minipage}\end{lrbox}%
  \begin{lrbox}{\sfboxc}\begin{minipage}{0.43\textwidth}
    \begin{minipage}[t]{0.49\linewidth}
      \[\begin{aligned}
        (A, 1) &\Delta (B, \{1\}) \\
        (A, 1) &\Delta (C, \{1\}) \\
        (B, 1) &\Delta (\langle D, \alpha \rangle, \{1\}) \\
        (B, 1) &\Delta (\langle D, \beta \rangle, \{1\}) \\
        (C, 1) &\Delta (\langle D, \beta \rangle, \{1\})
      \end{aligned}\]
    \end{minipage}%
    \hfill
    \begin{minipage}[t]{0.49\linewidth}
      \[\begin{aligned}
        (\langle D, \alpha \rangle, 1) &\Delta (D, \{1\}) \\
        (\langle D, \alpha \rangle, 1) &\Delta (E, \{1\}) \\
        (\langle D, \beta \rangle, 1) &\Delta (D, \{1\}) \\
        (\langle D, \beta \rangle, 1) &\Delta (F, \{1\})
      \end{aligned}\]
    \end{minipage}
  \end{minipage}\end{lrbox}%
  \begin{lrbox}{\sfboxd}\begin{minipage}{0.43\textwidth}\centering
    \includegraphics[scale=0.8]{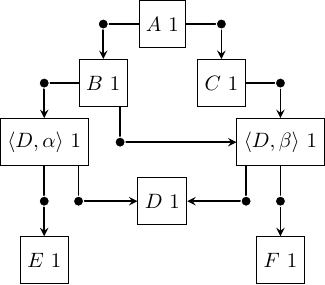}%
  \end{minipage}\end{lrbox}%
  \setlength{\sfrowht}{0pt}\sfmax{\sfboxa}\sfmax{\sfboxc}%
  \setlength{\sfrowhtb}{0pt}\sfmaxb{\sfboxb}\sfmaxb{\sfboxd}%
  \hfill
  \begin{subfigure}[b]{0.20\textwidth}
    \begin{minipage}[c][\dimexpr\sfrowht+\sfrowhtb+2\baselineskip\relax][c]{\linewidth}
{\scriptsize\bfseries\ttfamily A/Cargo.toml}
\begin{Verbatim}[fontsize=\scriptsize]
[dependencies]
B = "1"
C = "1"
\end{Verbatim}
{\scriptsize\bfseries\ttfamily B/Cargo.toml}
\begin{Verbatim}[fontsize=\scriptsize]
[dependencies.D]
version = "1"
features = [
  "alpha", "beta"]
\end{Verbatim}
{\scriptsize\bfseries\ttfamily C/Cargo.toml}
\begin{Verbatim}[fontsize=\scriptsize]
[dependencies.D]
version = "1"
features = ["beta"]
\end{Verbatim}
{\scriptsize\bfseries\ttfamily D/Cargo.toml}
\begin{Verbatim}[fontsize=\scriptsize]
[features]
alpha = ["dep:E"]
beta = ["dep:F"]
[dependencies]
E.version = "1"
E.optional = true
F.version = "1"
F.optional = true
\end{Verbatim}
    \end{minipage}
    \caption{Cargo syntax.}
    \label{fig.features.cargo}
  \end{subfigure}%
  \hfill
  \begin{subfigure}[b]{0.30\textwidth}
    \sfcell{\usebox{\sfboxa}}
    \begin{minipage}[c][2\baselineskip][c]{\linewidth}
      \caption{Feature Package Calculus.}
      \label{fig.features.calculus}
    \end{minipage}
    \sfcellb{\usebox{\sfboxb}}
    \Description{Hypergraph rendering of the feature calculus dependencies given alongside.}
    \caption{Hypergraph illustration.}
    \label{fig.features.hypergraph}
  \end{subfigure}%
  \hfill
  \begin{subfigure}[b]{0.43\textwidth}
    \sfcell{\usebox{\sfboxc}}
    \begin{minipage}[c][2\baselineskip][c]{\linewidth}
      \caption{Reduction to the core calculus.}
      \label{fig.feature-reduction.reduction}
    \end{minipage}
    \sfcellb{\usebox{\sfboxd}}
    \Description{Hypergraph rendering of the reduction formula given alongside.}
    \caption{Hypergraph of reduction.}
    \label{fig.feature-reduction.hypergraph}
  \end{subfigure}%
  \hfill\null
  \caption{Feature Package Calculus instance reduced to the core calculus.}
  \label{fig.features}
\end{figure}

\newlength{\opsubwd}
\begin{definition}[Feature Reduction]
  \label{def.feature-reduction}
  Given $R_f$, $\Delta_f$, $\Delta_a$, $\mathit{support}$, and $r_f$, we define a reduction to the core calculus with $R$, $\Delta$, and $r$.
  A dependee $(m, vs, fs)$ reduces to a set of \textit{targets} -- the base name when it requires no features, or its feature packages otherwise:
  \[
    t(m, vs, fs) =
    \begin{cases}
      \{ (m,\, vs) \} & \text{if } fs = \emptyset \\
      \{ (\langle m, f' \rangle,\, vs) \mid f' \in fs \} & \text{if } fs \neq \emptyset
    \end{cases}
  \]
  \begin{subdefinition}
    \item\label{def.feature-reduction.packages}
    \textbf{Packages}: $r = r_f$, and each supported feature gives a feature package name $\langle n, f \rangle \in N$:
    \[
      R = R_f \cup \bigcup_{(n, v) \in R_f} \{ (\langle n, f \rangle, v) \mid ((n, v), f) \in \mathit{support} \}
    \]
    \item\label{def.feature-reduction.dependencies}
    \textbf{Dependencies}: each synthetic feature package depends on its base package; each parameterised or additional dependency reduces to one dependency per target:
    \settowidth{\opsubwd}{$\scriptstyle ((n, v), f) \Delta_a (m, vs, fs)$}%
    \[
      \begin{aligned}
        \Delta = {}& \bigcup_{\mathmakebox[\opsubwd]{(n, v) \in R_f}} \{ ((\langle n, f \rangle, v),\, (n,\, \{v\})) \mid ((n, v), f) \in \mathit{support} \} \\
        {}\cup{}& \bigcup_{\mathmakebox[\opsubwd]{p \Delta_f (n, vs, fs)}} \{ (p,\, d) \mid d \in t(n, vs, fs) \} \\
        {}\cup{}& \bigcup_{\mathmakebox[\opsubwd]{((n, v), f) \Delta_a (m, vs, fs)}} \{ ((\langle n, f \rangle, v),\, d) \mid d \in t(m, vs, fs) \}
      \end{aligned}
    \]
  \end{subdefinition}
\end{definition}

\begin{theorem}[Soundness]
  \label{thm.feature-reduction-soundness}
  If $S \in \mathcal{S}(\Delta, r)$, then
    \[
      S_f = \{ ((n, v),\, \{ f \mid (\langle n, f \rangle, v) \in S \}) \mid (n, v) \in S \cap R_f \}
    \]
  is a valid resolution in $\mathcal{S}_f(\Delta_f, \Delta_a, r_f)$.
\end{theorem}

\begin{theorem}[Completeness]
  \label{thm.feature-reduction-completeness}
  If $S_f \in \mathcal{S}_f(\Delta_f, \Delta_a, r_f)$, then
    \[
      S = \bigcup_{((n, v), fs) \in S_f} \Big( \{ (n, v) \} \cup \{ (\langle n, f \rangle, v) \mid f \in fs \} \Big)
    \]
  is a valid resolution in $\mathcal{S}(\Delta, r)$.
\end{theorem}

\subsection{Package Formulae}
\label{sec.mise-en-place.package-formula}

The core's dependency relation is conjunctive: every dependency must be satisfied.
Some packaging DSLs allow Boolean formulae~(\S\ref{sec.background.survey.package-formula}) which a conjunction of dependencies alone cannot express.

\begin{definition}[Package Formula]
  \label{def.package-formula}
  We define:
  \begin{subdefinition}
    \item\label{def.package-formula.grammar}
    The set of package formulae $\Psi$ is given by the grammar:
    \[\psi ::= (n, vs) \mid \psi \land \psi \mid \psi \lor \psi \mid \neg\psi \qquad n \in N,\ vs \subseteq V\]
    \item\label{def.package-formula.satisfaction}
    The satisfaction relation $S_\Psi \models \psi$ (resolution $S_\Psi$ satisfies formula $\psi$) is defined:
    \begin{gather*}
      \inferrule*[right=DEP]{\exists\, v \in vs.\, (n, v) \in S_\Psi}{S_\Psi \models (n, vs)} \qquad
      \inferrule*[right=AND]{S_\Psi \models \psi_L \\ S_\Psi \models \psi_R}{S_\Psi \models \psi_L \land \psi_R} \qquad
      \inferrule*[right=NOT]{S_\Psi \not\models \psi}{S_\Psi \models \neg \psi} \\[1.5ex]
      \inferrule*[right=OR-L]{S_\Psi \models \psi_L}{S_\Psi \models \psi_L \lor \psi_R} \qquad
      \inferrule*[right=OR-R]{S_\Psi \models \psi_R}{S_\Psi \models \psi_L \lor \psi_R}
    \end{gather*}
  \end{subdefinition}
\end{definition}

\begin{definition}[Package Formula Dependency]
  \label{def.package-formula-dependency}
  We define package formula dependencies $\Delta_\Psi \subseteq (N \times V) \times \Psi$ as a relation where an element $p \Delta_\Psi \psi$ denotes that package $p$ depends on formula $\psi$.
\end{definition}

\begin{definition}[Package Formula Resolution]
  \label{def.package-formula-resolution}
  Given dependencies $\Delta_\Psi$ and root $r_\Psi$, a resolution $S_\Psi \subseteq R_\Psi$ is valid if:
  \begin{subdefinition}
    \item\label{def.package-formula-resolution.root-inclusion}
    \textbf{Root Inclusion}: $r_\Psi \in S_\Psi$, as in Def.~\ref{def.calculus.resolution.root-inclusion}.
    \item\label{def.package-formula-resolution.formula-closure}
    \textbf{Formula Closure}: $\forall\, p \in S_\Psi.\, p \Delta_\Psi \psi \implies S_\Psi \models \psi$
    \item\label{def.package-formula-resolution.version-uniqueness}
    \textbf{Version Uniqueness}: $\forall\, (n, v),\, (n, v') \in S_\Psi.\, v = v'$, as in Def.~\ref{def.calculus.resolution.version-uniqueness}.
  \end{subdefinition}
  We write $\mathcal{S}_\Psi(\Delta_\Psi, r_\Psi)$ for the set of all resolutions of $r_\Psi$ in $\Delta_\Psi$.
\end{definition}

We reduce the Package Formula Calculus to the core calculus using synthetic packages for disjunction and the conflict encoding from Def.~\ref{def.conflict-reduction} for negation.
Each disjunction node introduces one synthetic package name with two versions (one per disjunct), analogous to a Tseitin auxiliary variable~\cite{tseitin1968complexity}, keeping the reduction linear in the formula size and avoiding the exponential blowup of na\"{i}ve conversion to conjunctive normal form (CNF).
Fig.~\ref{fig.package-formula} shows an instance reduced to the core.

\begin{figure}[ht]
  \captionsetup{justification=centering}
  \begin{lrbox}{\sfboxa}\begin{minipage}{0.17\textwidth}%
\begin{Verbatim}[fontsize=\footnotesize]
# A-1.ebuild
DEPEND="|| (
  ( =B-2 =C-1 )
  ( =B-1 !!=C-1 )
)"
\end{Verbatim}
  \end{minipage}\end{lrbox}%
  \begin{lrbox}{\sfboxb}\begin{minipage}{0.69\textwidth}\centering
    \includegraphics[scale=0.8]{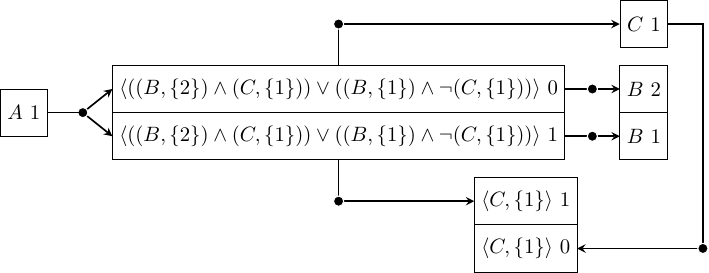}%
  \end{minipage}\end{lrbox}%
  \setlength{\sfrowht}{0pt}\sfmax{\sfboxa}\sfmax{\sfboxb}%
  \hfill
  \begin{subfigure}[b]{0.17\textwidth}
    \sfcell{\usebox{\sfboxa}}
    \caption{Portage syntax.}
    \label{fig.package-formula.portage}
  \end{subfigure}%
  \hfill
  \begin{subfigure}[b]{0.69\textwidth}
    \sfcell{\usebox{\sfboxb}}
    \Description{Hypergraph rendering of the package-formula reduction whose source is given in the caption.}
    \caption{Hypergraph of reduction.}
    \label{fig.package-formula.hypergraph-reduction}
  \end{subfigure}%
  \hfill\null
  \caption{Package Formula Calculus instance $(A, 1) \Delta_\Psi (((B, \{2\}) \land (C, \{1\})) \lor ((B, \{1\}) \land \neg (C, \{1\})))$ reduced to the core calculus.}
  \label{fig.package-formula}
\end{figure}

\pagebreak

\begin{definition}[Package Formula Reduction]
  \label{def.package-formula-reduction}
  Given $R_\Psi$, $\Delta_\Psi$, and $r_\Psi$, we define a reduction to the core calculus with $R$, $\Delta$, and $r$ as follows:
  \begin{subdefinition}
    \item\label{def.package-formula-reduction.packages}
    \textbf{Packages}: $r = r_\Psi$, and synthetic packages for each disjunction and negated atom arising in the encoding $\mathcal{E}$ (Def.~\ref{def.package-formula-reduction.formula}), including those introduced by its De Morgan rewriting:
    \[
      R = R_\Psi
      \cup \bigcup_{\substack{\psi_L \lor \psi_R \\ \text{arising in } \mathcal{E}}} \{ (\langle \psi_L \lor \psi_R \rangle, 0),\, (\langle \psi_L \lor \psi_R \rangle, 1) \}
      \cup \bigcup_{\substack{\neg (n, vs) \\ \text{arising in } \mathcal{E}}} \{ (\langle n, vs \rangle, 0),\, (\langle n, vs \rangle, 1) \}
    \]
    where $\langle \psi_L \lor \psi_R \rangle, \langle n, vs \rangle \in N$ and $0, 1 \in V$.
    \item\label{def.package-formula-reduction.formula}
    \textbf{Formula}:
    we define the encoding function $\mathcal{E}: \Delta_\Psi \to \mathcal{P}(\Delta)$:
    \[
      \begin{aligned}
        \mathcal{E}(p \Delta_\Psi (n, vs)) &= \{ (p, (n, vs)) \} \\[4pt]
        \mathcal{E}(p \Delta_\Psi (\psi_L \land \psi_R)) &= \mathcal{E}(p \Delta_\Psi \psi_L) \cup \mathcal{E}(p \Delta_\Psi \psi_R) \\[4pt]
        \mathcal{E}(p \Delta_\Psi (\psi_L \lor \psi_R)) &= \{ (p, (i, \{0, 1\})) \} \cup \mathcal{E}((i, 0) \Delta_\Psi \psi_L) \cup \mathcal{E}((i, 1) \Delta_\Psi \psi_R)  \text{ with } i \coloneqq \langle \psi_L \lor \psi_R \rangle \\[4pt]
        \mathcal{E}(p \Delta_\Psi \neg (n, vs)) &= \{ (p, (\langle n, vs \rangle, \{1\})) \} \cup \{ ((n, u), (\langle n, vs \rangle, \{0\})) \mid u \in vs \} \\[4pt]
        \mathcal{E}(p \Delta_\Psi \neg (\psi_L \land \psi_R)) &= \mathcal{E}(p \Delta_\Psi (\neg \psi_L \lor \neg \psi_R)) \\[4pt]
        \mathcal{E}(p \Delta_\Psi \neg (\psi_L \lor \psi_R)) &= \mathcal{E}(p \Delta_\Psi (\neg \psi_L \land \neg \psi_R)) \\[4pt]
        \mathcal{E}(p \Delta_\Psi \neg \neg \psi) &= \mathcal{E}(p \Delta_\Psi \psi)
      \end{aligned}
    \]
   \item\label{def.package-formula-reduction.dependencies}
   \textbf{Dependencies}: we apply the encoding function to each dependency $\displaystyle \Delta = \bigcup_{p \Delta_\Psi \psi} \mathcal{E}(p \Delta_\Psi \psi)$.
  \end{subdefinition}
\end{definition}

\begin{theorem}[Soundness]
  \label{thm.package-formula-reduction-soundness}
  If $S \in \mathcal{S}(\Delta, r)$, then
  $S_\Psi = S \cap R_\Psi$
  is a valid resolution in $\mathcal{S}_\Psi(\Delta_\Psi, r_\Psi)$.
\end{theorem}

\begin{theorem}[Completeness]
  \label{thm.package-formula-reduction-completeness}
  If $S_\Psi \in \mathcal{S}_\Psi(\Delta_\Psi, r_\Psi)$, define the witness set $W^c(\psi)$ for $c \in \{\top, \bot\}$, where $\top$ marks subformulae on a \textit{taken} path, and $\bot$ on an \textit{untaken} path:
  \begin{align*}
    W^c((n, vs)) &= \emptyset \\[4pt]
    W^c(\psi_L \land \psi_R) &= W^c(\psi_L) \cup W^c(\psi_R) \\[4pt]
    W^\top(\neg(n, vs)) &= \{(\langle n, vs \rangle, 1)\} \\[4pt]
    W^\bot(\neg(n, vs)) &= \begin{cases}\{(\langle n, vs \rangle, 0)\} & \text{if } \exists\, u \in vs.\, (n, u) \in S_\Psi \\ \emptyset & \text{otherwise}\end{cases} \\[4pt]
    W^\top(\psi_L \lor \psi_R) &= \begin{cases}\{(\langle \psi_L \lor \psi_R \rangle, 0)\} \cup W^\top(\psi_L) \cup W^\bot(\psi_R) & \text{if } S_\Psi \models \psi_L \\ \{(\langle \psi_L \lor \psi_R \rangle, 1)\} \cup W^\bot(\psi_L) \cup W^\top(\psi_R) & \text{otherwise}\end{cases} \\[4pt]
    W^\bot(\psi_L \lor \psi_R) &= W^\bot(\psi_L) \cup W^\bot(\psi_R) \\[4pt]
    W^c(\neg(\psi_L \land \psi_R)) &= W^c(\neg\psi_L \lor \neg\psi_R) \\[4pt]
    W^c(\neg(\psi_L \lor \psi_R)) &= W^c(\neg\psi_L \land \neg\psi_R) \\[4pt]
    W^c(\neg\neg\psi) &= W^c(\psi)
  \end{align*}
  Then $\displaystyle S = S_\Psi \cup \bigcup_{p \Delta_\Psi \psi} \begin{cases} W^\top(\psi) & \text{if } p \in S_\Psi \\ W^\bot(\psi) & \text{otherwise} \end{cases}$ is a valid resolution in $\mathcal{S}(\Delta, r)$.
\end{theorem}

\subsection{Variable Formulae}
\label{sec.mise-en-place.variable-formula}

Packaging DSLs sometimes allow variables in package formulae~(\S\ref{sec.background.survey.variable-formula}); we extend the Boolean package formula with variables, values, and comparisons.

\begin{definition}[Variable Formula]
  \label{def.variable-formula}
  Extending the Package Formula (Def.~\ref{def.package-formula}), we define:
  \begin{subdefinition}
    \item\label{def.variable-formula.variables}
    The finite set of variables $X$, the non-empty set of variable values $Y$ totally ordered by $\leq_Y$, the non-empty set $Y_x$ of possible values for each $x \in X$, and assignment function $\sigma : (x : X) \to Y_x$.
    \item\label{def.variable-formula.grammar}
    The set of variable formulae $\Psi$ is given by the grammar:
    \[\psi ::= (n, vs) \mid \psi \land \psi \mid \psi \lor \psi \mid \neg\psi \mid x\ \omega\ y \qquad \omega ::=\ \geq\ \mid\ >\ \mid\ \leq\ \mid\ <\ \mid\ =\ \mid\ \neq \qquad x \in X,\ y \in Y\]
    \item\label{def.variable-formula.dependencies}
    Variable formula dependencies $\Delta_\Psi \subseteq (N \times V) \times \Psi$ as in Def.~\ref{def.package-formula-dependency}.
    \item\label{def.variable-formula.satisfaction}
    The satisfaction relation $(S_\Psi, \sigma) \models \psi$ (resolution $S_\Psi$ and assignment $\sigma$ satisfy formula $\psi$) extends Def.~\ref{def.package-formula.satisfaction} with the judgement form carrying assignment $\sigma$.
    Writing $\omega_Y$ for the interpretation of $\omega$ under $\leq_Y$, we define:
    \[\inferrule*{\sigma(x)\ \omega_Y\ y}{(S_\Psi, \sigma) \models (x\ \omega\ y)}\]
  \end{subdefinition}
\end{definition}

\begin{definition}[Variable Formula Resolution]
  \label{def.variable-formula-resolution}
  Given dependencies $\Delta_\Psi$ and root $r_\Psi$, a resolution $S_\Psi \subseteq R_\Psi$ and assignment $\sigma$ are valid if:
  \begin{subdefinition}
    \item\label{def.variable-formula-resolution.root-inclusion}
    \textbf{Root Inclusion}: $r_\Psi \in S_\Psi$, as in Def.~\ref{def.calculus.resolution.root-inclusion}.
    \item\label{def.variable-formula-resolution.formula-closure}
    \textbf{Formula Closure}: $\forall\, p \in S_\Psi.\, p \Delta_\Psi \psi \implies (S_\Psi, \sigma) \models \psi$
    \item\label{def.variable-formula-resolution.version-uniqueness}
    \textbf{Version Uniqueness}: $\forall\, (n, v),\, (n, v') \in S_\Psi.\, v = v'$, as in Def.~\ref{def.calculus.resolution.version-uniqueness}.
  \end{subdefinition}
  We write $\mathcal{S}_\Psi(\Delta_\Psi, r_\Psi) \ni (S_\Psi, \sigma)$ for all resolution-assignment pairs of $r_\Psi$ in $\Delta_\Psi$.
\end{definition}

Variable formulae can capture package-level platform constraints.
For example, in npm, \texttt{"os": ["linux"]} depends on variable `os' with value `linux'.
A filtered dependency in opam, \texttt{depends: ["foo" \{os = "linux"\}]}, becomes $\neg(\mathit{os} = \mathit{linux}) \lor (\mathit{foo}, V)$.
Package-local variables, such as opam's \texttt{with-test}, are namespaced by the package.
Some package managers perform this expansion as pre-processing, but lifting it into resolution lets variables be solved for directly -- for example, choosing the Linux distribution with the freshest system dependencies.

To reduce to the core, we extend the Package Formula Reduction (Def.~\ref{def.package-formula-reduction}).

\begin{definition}[Variable Formula Reduction]
  \label{def.variable-formula-reduction}
  Given $R_\Psi$, $\Delta_\Psi$, and $r_\Psi$, we define a reduction to the core calculus with $R$, $\Delta$, and $r$ as follows:
  \begin{subdefinition}
    \item\label{def.variable-formula-reduction.packages}
    \textbf{Packages}: as Def.~\ref{def.package-formula-reduction.packages}, with
    $R \supseteq \{ (\langle x \rangle, y) \mid x \in X,\ y \in Y_x \}$ where $\langle x \rangle \in N$ and $y \in V$.
    \item\label{def.variable-formula-reduction.formula}
    \textbf{Formula}:
    we extend the encoding function $\mathcal{E}$ (Def.~\ref{def.package-formula-reduction.formula}) with:
    \[
        \mathcal{E}(p \Delta_\Psi x\ \omega\ y ) = \{ (p, (\langle x \rangle, \{v \in Y_x \mid v\ \omega_Y\ y\})) \} \qquad
        \mathcal{E}(p \Delta_\Psi \neg(x\ \omega\ y) ) = \mathcal{E}(p \Delta_\Psi x\ \bar\omega\ y )
    \]
   where $\bar{\omega}$ denotes the complement: $\bar{\geq}$ is ${<}$, $\bar{>}$ is ${\leq}$, $\bar{=}$ is ${\neq}$, etc.
   \item\label{def.variable-formula-reduction.dependencies}
   \textbf{Dependencies}: $\displaystyle \Delta = \bigcup_{p \Delta_\Psi \psi} \mathcal{E}(p \Delta_\Psi \psi)$ as in Def.~\ref{def.package-formula-reduction.dependencies}.
  \end{subdefinition}
\end{definition}

\begin{theorem}[Soundness]
  \label{thm.variable-formula-reduction-soundness}
  If $S \in \mathcal{S}(\Delta, r)$, then
  $S_\Psi = S \cap R_\Psi$ and $\sigma(x) = y$ where $(\langle x \rangle, y) \in S$, arbitrary otherwise,
  are a valid resolution and assignment in $\mathcal{S}_\Psi(\Delta_\Psi, r_\Psi)$.
\end{theorem}

\begin{theorem}[Completeness]
  \label{thm.variable-formula-reduction-completeness}
  If $(S_\Psi, \sigma) \in \mathcal{S}_\Psi(\Delta_\Psi, r_\Psi)$, define the witness set $W^c(\psi)$ for $c \in \{\top, \bot\}$ as in Thm.~\ref{thm.package-formula-reduction-completeness}, extended with $W^c(x\ \omega\ y) = W^c(\neg(x\ \omega\ y)) = \emptyset$ and using $(S_\Psi, \sigma) \models \psi$ for satisfaction tests.

  Then
  $\displaystyle S = S_\Psi \cup \bigcup_{p \Delta_\Psi \psi} \begin{cases} W^\top(\psi) & \text{if } p \in S_\Psi \\ W^\bot(\psi) & \text{otherwise} \end{cases} \cup \{ (\langle x \rangle, \sigma(x)) \mid x \in X \}$
  is a valid resolution in $\mathcal{S}(\Delta, r)$.
\end{theorem}

\subsection{Virtual Packages}
\label{sec.mise-en-place.virtual}

Some package managers support dependencies on \textit{virtual packages}~(\S\ref{sec.background.survey.virtual-packages}) that do not exist as real packages, but are instead provided by a real package.
We extend the core with a provides relation.

\begin{definition}[Virtual Package Provides]
  \label{def.virtual-package-provides}
  We define provides $\Pi \subseteq (N \times V) \times (N \times (V \cup \{\top\}))$ as a relation where an element $p \Pi (n, v)$ denotes that package $p$ provides a package name $n \in N$ with version $v \in V \cup \{\top\}$, where $\top$ matches any version.
  We define $v \in_\top vs \iff v \in vs \cup \{\top\}$.
  The name $n$ is \textit{real} if it appears in $R_\Pi$, and is otherwise \textit{virtual}.
\end{definition}

We write $\Delta_\Pi$ for dependencies (Def.~\ref{def.calculus.dependency}); a dependee name need not be real, since a virtual name is satisfied by a provider rather than a real package.

\begin{definition}[Virtual Package Resolution]
  \label{def.virtual-package-resolution}
  Given dependencies $\Delta_\Pi$, provides $\Pi$, and root $r_\Pi$, a resolution $S_\Pi \subseteq R_\Pi$ and provider relation $\rho \subseteq S_\Pi \times N \times S_\Pi$ -- read $((m, u), n, p) \in \rho$ as $(m, u)$ is $p$'s provider for name $n$ -- are valid if:
  \begin{subdefinition}
    \item\label{def.virtual-package-resolution.root-inclusion}
    \textbf{Root Inclusion}: $r_\Pi \in S_\Pi$, as in Def.~\ref{def.calculus.resolution.root-inclusion}.
    \item\label{def.virtual-package-resolution.virtual-dependency-closure}
    \textbf{Virtual Dependency Closure}:
    \begin{alignat*}{2}
      \forall\, p \in S_\Pi.\, p \Delta_\Pi (n, vs) \implies
      &\big(\exists\, v \in vs.\, (n, v) \in S_\Pi\big) \\
      \lor\ &\big(\exists!\, (m, u) \in S_\Pi.\, \exists\, v \in_\top vs.\, (m, u) \Pi (n, v) \land ((m, u), n, p) \in \rho\big)
    \end{alignat*}
    Each dependency is satisfied by a compatible package or a unique provider, witnessed by $\rho$.
    \item\label{def.virtual-package-resolution.version-uniqueness}
    \textbf{Version Uniqueness}: $\forall\, (n, v),\, (n, v') \in S_\Pi.\, v = v'$, as in Def.~\ref{def.calculus.resolution.version-uniqueness}.
  \end{subdefinition}
  We write $\mathcal{S}_\Pi(\Delta_\Pi, \Pi, r_\Pi) \ni (S_\Pi, \rho)$ for all resolution-provider pairs of $r_\Pi$ in $\Delta_\Pi$ and $\Pi$.
\end{definition}

We reduce the Virtual Package Calculus to the core calculus using intermediate package names with versions that encode which provider is selected.
As with the parent relation~(\S\ref{sec.mise-en-place.concurrent}), the provider relation $\rho$ and the intermediate packages of the reduction below determine one another (Thms~\ref{thm.virtual-package-reduction-soundness} and~\ref{thm.virtual-package-reduction-completeness}).
Fig.~\ref{fig.virtual} shows an instance of the Virtual Package Calculus reduced to the core.

\begin{definition}[Virtual Package Reduction]
  \label{def.virtual-package-reduction}
  Given $R_\Pi$, $\Delta_\Pi$, $\Pi$, and $r_\Pi$, we define a reduction to the core calculus with $R$, $\Delta$, and $r$.
  For each dependency $p \Delta_\Pi (n, vs)$, define a predicate for whether the dependee has a provider as
  \[\mathit{prov}(n, vs) \iff \exists\, q,\, v \in_\top vs.\, q \Pi (n, v)\]
  and collect versions of the intermediate name $\langle p, n\rangle$ as
  \[us = \{\langle m, w \rangle \mid (m, w) \Pi (n, v),\ v \in_\top vs\} \cup \{\langle n, u \rangle \mid u \in vs,\ (n, u) \in R_\Pi\}\]
  with $\langle p, n \rangle \in N$ and $\langle m, w \rangle, \langle n, u \rangle \in V$.
  \begin{subdefinition}
    \item\label{def.virtual-package-reduction.packages}
    \textbf{Packages}: $r = r_\Pi$, and each package that can satisfy a dependency on $n$ -- a provider of $n$, or a real package named $n$ -- is exposed as a version of the intermediate name $\langle p, n \rangle$:
    \[
      R = R_\Pi \cup \bigcup_{\substack{p \Delta_\Pi (n, vs) \\ \mathit{prov}(n, vs)}} \{ (\langle p, n \rangle, w) \mid w \in us \}
    \]
    \item\label{def.virtual-package-reduction.dependencies}
    \textbf{Dependencies}: a dependency with no provider reduces to a direct core dependency; otherwise to a dependency on the intermediate name $\langle p, n \rangle$, whose versions each depend on the corresponding provider or real package:
    \[
      \Delta = \bigcup_{p \Delta_\Pi (n, vs)}
      \begin{cases}
        \{ (p,\, (n,\, vs)) \} & \text{if } \neg\mathit{prov}(n, vs) \\[6pt]
        \{ (p,\, (\langle p, n \rangle,\, us)) \}
        \cup \{ ((\langle p, n \rangle, \langle m, w \rangle),\, (m,\, \{w\})) \mid \langle m, w \rangle \in us \} & \text{otherwise}
      \end{cases}
    \]
  \end{subdefinition}
\end{definition}

\begin{figure}[ht]
  \captionsetup{justification=centering}
  \begin{lrbox}{\sfboxa}\begin{minipage}{0.32\textwidth}\centering
    \begin{align*}
      (B, 1) &\Pi (V, 1) \\
      (C, 1) &\Pi (V, 1) \\
      (A, 1) &\Delta_\Pi (V, \{1\})
    \end{align*}
  \end{minipage}\end{lrbox}%
  \begin{lrbox}{\sfboxb}\begin{minipage}{0.32\textwidth}\centering
    \includegraphics[scale=0.8]{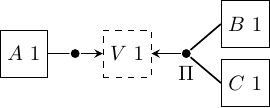}%
  \end{minipage}\end{lrbox}%
  \begin{lrbox}{\sfboxc}\begin{minipage}{0.40\textwidth}\centering
    \begin{align*}
      &(A, 1) \Delta (\langle (A, 1), V \rangle, \{\langle B, 1 \rangle, \langle C, 1 \rangle\}) \\
      &(\langle (A, 1), V \rangle, \langle B, 1 \rangle) \Delta (B, \{1\}) \\
      &(\langle (A, 1), V \rangle, \langle C, 1 \rangle) \Delta (C, \{1\})
    \end{align*}
  \end{minipage}\end{lrbox}%
  \begin{lrbox}{\sfboxd}\begin{minipage}{0.40\textwidth}\centering
    \includegraphics[scale=0.8]{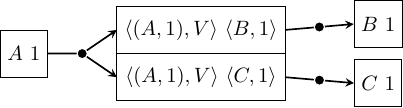}%
  \end{minipage}\end{lrbox}%
  \setlength{\sfrowht}{0pt}\sfmax{\sfboxa}\sfmax{\sfboxc}%
  \setlength{\sfrowhtb}{0pt}\sfmaxb{\sfboxb}\sfmaxb{\sfboxd}%
  \hfill
  \begin{subfigure}[b]{0.20\textwidth}
    \begin{minipage}[c][\dimexpr\sfrowht+\sfrowhtb+2\baselineskip\relax][c]{\linewidth}
\begin{Verbatim}[fontsize=\footnotesize]
Package: A
Version: 1
Depends: V (= 1)

Package: B
Version: 1
Provides: V (= 1)

Package: C
Version: 1
Provides: V (= 1)
\end{Verbatim}
    \end{minipage}
    \caption{Debian syntax.}
    \label{fig.virtual.debian}
  \end{subfigure}%
  \hfill
  \begin{subfigure}[b]{0.32\textwidth}
    \sfcell{\usebox{\sfboxa}}
    \begin{minipage}[c][2\baselineskip][c]{\linewidth}
      \caption{Virtual Package Calculus.}
      \label{fig.virtual.calculus}
    \end{minipage}
    \sfcellb{\usebox{\sfboxb}}
    \Description{Hypergraph rendering of the virtual calculus dependencies given alongside.}
    \caption{Hypergraph illustration.}
    \label{fig.virtual.hypergraph}
  \end{subfigure}%
  \hfill
  \begin{subfigure}[b]{0.40\textwidth}
    \sfcell{\usebox{\sfboxc}}
    \begin{minipage}[c][2\baselineskip][c]{\linewidth}
      \caption{Reduction to the core.}
      \label{fig.virtual.reduction}
    \end{minipage}
    \sfcellb{\usebox{\sfboxd}}
    \Description{Hypergraph rendering of the reduction formula given alongside.}
    \caption{Hypergraph of reduction.}
    \label{fig.virtual.hypergraph-reduction}
  \end{subfigure}%
  \hfill\null
  \caption{Virtual Package Calculus instance reduced to the core calculus.}
  \label{fig.virtual}
\end{figure}

\begin{theorem}[Soundness]
  \label{thm.virtual-package-reduction-soundness}
  If $S \in \mathcal{S}(\Delta, r)$, then
  $S_\Pi = S \cap R_\Pi$ and
  \[
    \rho = \bigcup_{p \Delta_\Pi (n, vs)} \left\{ ((m, w),\, n,\, p) \;\middle|\;
    \begin{aligned}
      &p \in S_\Pi,\ (m, w) \Pi (n, v),\ v \in_\top vs, \\
      &(\langle p, n \rangle, \langle m, w \rangle) \in S
    \end{aligned}
    \right\}
  \]
  are a valid resolution and provider relation in $\mathcal{S}_\Pi(\Delta_\Pi, \Pi, r_\Pi)$.
\end{theorem}

\begin{theorem}[Completeness]
  \label{thm.virtual-package-reduction-completeness}
  If $(S_\Pi, \rho) \in \mathcal{S}_\Pi(\Delta_\Pi, \Pi, r_\Pi)$, then
  \[
    S = S_\Pi \cup \bigcup_{\substack{p \Delta_\Pi (n, vs) \\ p \in S_\Pi,\ \mathit{prov}(n, vs)}} \{ (\langle p, n \rangle, c(p, n)) \}
  \]
  is a valid resolution in $\mathcal{S}(\Delta, r)$, where, given $p \Delta_\Pi (n, vs)$, $c(p, n)$ selects a provider through $\rho$ if possible and the direct version otherwise:
  \[
    c(p, n) =
    \begin{cases}
      \langle m, w \rangle & \text{for some } (m, w) \in S_\Pi \text{ with } \exists\, v \in_\top vs.\, (m, w) \Pi (n, v),\ ((m, w), n, p) \in \rho \\
      \langle n, u \rangle & \text{otherwise, with } u \in vs,\ (n, u) \in S_\Pi
    \end{cases}
  \]
\end{theorem}

\section{Package Managers, \`a la Carte}
\label{sec.a-la-carte}

We return to the motivating example of~\S\ref{sec.introduction}, using four package managers for a single project.
Dependencies between these ecosystems are unversioned; the binary will fail if APT has not been used to install the required C libraries and drivers from the Debian repository at compatible versions.
For instance, the GPU bindings pip installs must match the kernel driver APT provides -- a constraint neither resolver can express.
Build sandboxing prevents one package manager from invoking another, so opam cannot call Cargo during a sandboxed build.
This leads to ad hoc solutions, such as opam invoking a system package manager to install dependencies via its depext mechanism~\cite{opam}.
Another workaround is vendoring dependees across ecosystems -- copying the source code into a project's version control system -- so that versions are fixed and dependencies are available inside sandboxed build environments.
Similarly, containerisation technologies like Docker~\cite{madhavapeddy2025docker,madhavapeddy2026decade} are used to sidestep the problem by bundling all dependencies together.
To avoid vendoring or bundling in containers, packages can be duplicated between ecosystem-specific package formats, such as Rust projects packaged in opam.
But these workarounds carry serious costs: vendored and duplicated packages require ongoing maintenance when new versions are released; and the dependencies they introduce are hidden from package managers, rendering security vulnerability tracking incomplete.

The Package Calculus offers an alternative.
As the core calculus~(\S\ref{sec.calculus.core}) is expressive enough to model each axis of divergence in dependency expression~(\S\ref{sec.mise-en-place}), a project spanning multiple ecosystems can be resolved by reducing each ecosystem's dependencies to the core and solving the combined instance with a single \textit{polyglot resolver}.
A dependency that crosses between ecosystems is then an ordinary core dependency.
Realising this requires modelling real package managers, which combine several extensions rather than using them in isolation; we show how to compose calculi and their reductions~(\S\ref{sec.a-la-carte.composition}).

More ambitiously, the core could serve as an intermediate representation for \textit{translating} dependencies from one ecosystem's language into another, routing $n^2$ direct translators through $2n$ reductions and liftings~(\S\ref{sec.a-la-carte.transpiling}).
This direction is more speculative: it requires inverting the reductions.
The reduction of a single-extension instance to the core can be inverted: decoding the synthetic names recovers an equivalent extended instance.
Inverting the encoding of composed extensions, whose synthetic structures interleave, follows no known general principle.

\subsection{Composition of Extensions}
\label{sec.a-la-carte.composition}

The reductions defined in~\S\ref{sec.mise-en-place} form a toolkit for constructing the calculi and reductions of combined extensions, to model real package managers.
Cargo, for example, uses both concurrent versions~(\S\ref{sec.mise-en-place.concurrent}) and features~(\S\ref{sec.mise-en-place.features}).

\subsubsection{Composing Features and Concurrent Versions}
\label{sec.a-la-carte.composition.features-concurrent}

Composing the two gives the \textit{Concurrent Feature Package Calculus}.
As in the Feature Package Calculus~(\S\ref{sec.mise-en-place.features}), each selected package is paired with a set of features, and as in the Concurrent Package Calculus~(\S\ref{sec.mise-en-place.concurrent}), distinct granularity classes of a name may coexist.
The conditions below merge the two -- a parameterised ($\Delta_f$) or additional ($\Delta_a$) dependency must select a unique granular version of its dependee, while the feature and granularity invariants of both extensions are preserved.

\begin{definition}[Concurrent Feature Resolution]
  \label{def.concurrent-feature-resolution}
  Given parameterised dependencies $\Delta_f$, additional dependencies $\Delta_a$, granularity $g$, and root $r_{CF}$ with no features, a resolution $S_{CF} \subseteq R_{CF} \times \mathcal{P}(F)$ and parent relation $\pi \subseteq (N \times V) \times (N \times V)$ are valid if:
  \begin{subdefinition}
    \item\label{def.concurrent-feature-resolution.root-inclusion}
    \textbf{Root Inclusion}: $(r_{CF}, \emptyset) \in S_{CF}$
    \smallskip
    \item\label{def.concurrent-feature-resolution.parent-closure}
    \textbf{Parameterised and Additional Parent Closure}: define
    \[\mathit{sel}(p, n, vs, fs) \iff \exists!\, v \in vs.\, \exists\, fs' \supseteq fs.\, ((n, v), fs') \in S_{CF} \land ((n, v), p) \in \pi\]
    a predicate for whether, when $p$ depends on $n$ with versions $vs$ and features $fs$, exactly one version of $n$ in $vs$ is selected with at least the required features, with the parent recorded in $\pi$.
    \begin{gather*}
      \forall\, (p, fs_p) \in S_{CF}.\, p \Delta_f (n, vs, fs) \implies \mathit{sel}(p, n, vs, fs) \\
      \forall\, (p, fs_p) \in S_{CF},\, f \in fs_p.\, (p, f) \Delta_a (n, vs, fs) \implies \mathit{sel}(p, n, vs, fs)
    \end{gather*}
    \item\label{def.concurrent-feature-resolution.parent-functional}
    \textbf{Parent Relation Functionality}:
    \[\forall\, ((n, v), p),\, ((n, v'), p) \in \pi.\, v = v'\]
    A depender's $\Delta_f$ and $\Delta_a$ dependencies must agree on a single version of a name.
    \smallskip
    \item\label{def.concurrent-feature-resolution.version-granularity}
    \textbf{Version Granularity}: $\forall\, ((n, v), fs),\, ((n, v'), fs') \in S_{CF}.\, v \neq v' \implies g(v) \neq g(v')$
    \smallskip
    \item\label{def.concurrent-feature-resolution.feature-unification}
    \textbf{Feature Unification}: $\forall\, ((n, v), fs),\, ((n, v), fs') \in S_{CF}.\, fs = fs'$
    \smallskip
    \item\label{def.concurrent-feature-resolution.support}
    \textbf{Support}: $\forall\, (p, fs) \in S_{CF},\, f \in fs.\, (p, f) \in \mathit{support}$
    \smallskip
  \end{subdefinition}
  We write $\mathcal{S}_{CF}(\Delta_f, \Delta_a, g, r_{CF}) \ni (S_{CF}, \pi)$ for all resolution-parent pairs of $r_{CF}$ in $\Delta_f$ and $\Delta_a$ under $g$.
\end{definition}

\begin{definition}[Concurrent Feature Reduction]
  \label{def.concurrent-feature-reduction}
  Given $R_{CF}$, $\Delta_f$, $\Delta_a$, $\mathit{support}$, $g$, and $r_{CF}$, we define a reduction to the core calculus with $R$, $\Delta$, and $r$ by combining the Feature Reduction (Def.~\ref{def.feature-reduction}) and Concurrent Reduction (Def.~\ref{def.concurrent-reduction}).
  For each $(n, v) \Delta_f (m, vs, fs)$ and $((n, v), f) \Delta_a (m, vs, fs)$, write $i = \langle n, v, m \rangle$ for the intermediate name; in the helpers below, the argument $j$ is a per-feature intermediate name in which $f'$ is bound by each $f' \in fs$ comprehension.
  \begin{subdefinition}
    \item\label{def.concurrent-feature-reduction.packages}
    \textbf{Packages}: $r = (\langle n, g(v) \rangle, v)$ where $r_{CF} = (n, v)$.
    Writing
    \[
      P(j) = \{ (i, u) \mid u \in vs \} \cup \{ (j, u) \mid f' \in fs,\ u \in vs \},
    \]
    the real packages are
    \settowidth{\opsubwd}{$\scriptstyle ((n, v), f) \Delta_a (m, vs, fs)$}%
    \[
      \begin{aligned}
        R = {}& \bigcup_{\mathmakebox[\opsubwd]{(n, v) \in R_{CF}}} \{ (\langle n, g(v) \rangle, v) \} & {}\cup{} & \bigcup_{\mathmakebox[\opsubwd]{\substack{(n, v) \in R_{CF} \\ ((n, v), f) \in \mathit{support}}}} \{ (\langle \langle n, f \rangle, g(v) \rangle, v) \} \\
        {}\cup{}& \bigcup_{\mathmakebox[\opsubwd]{(n, v) \Delta_f (m, vs, fs)}} P(\langle n, v, m, f' \rangle) & {}\cup{} & \bigcup_{\mathmakebox[\opsubwd]{((n, v), f) \Delta_a (m, vs, fs)}} P(\langle n, v, f, m, f' \rangle)
      \end{aligned}
    \]
    \item\label{def.concurrent-feature-reduction.dependencies}
    \textbf{Dependencies}: for a granular depender $p$ and per-feature intermediate name $j$,
    \[
      \begin{aligned}
        D(p, j) = {}& \{ (p,\, (i,\, vs)) \} {}\cup{} \{ ((i, u),\, (\langle m, g(u) \rangle,\, \{u\})) \mid u \in vs \} \\
        {}\cup{}& \bigcup_{f' \in fs} \left(
          \begin{aligned}
            & \{ (p,\, (j,\, vs)) \} \\
            {}\cup{}& \{ ((j, u),\, (\langle \langle m, f' \rangle, g(u) \rangle,\, \{u\})),\ ((j, u),\, (i,\, \{u\})) \mid u \in vs \}
          \end{aligned}
        \right)
      \end{aligned}
    \]
    the dependencies are
    \settowidth{\opsubwd}{$\scriptstyle ((n, v), f) \Delta_a (m, vs, fs)$}%
    \[
      \begin{aligned}
        \Delta = {}& \bigcup_{\mathmakebox[\opsubwd]{\substack{(n, v) \in R_{CF} \\ ((n, v), f) \in \mathit{support}}}}
          \{ ((\langle \langle n, f \rangle, g(v) \rangle, v),\, (\langle n, g(v) \rangle,\, \{v\})) \} \\
        {}\cup{}& \bigcup_{\mathmakebox[\opsubwd]{(n, v) \Delta_f (m, vs, fs)}} D\big((\langle n, g(v) \rangle, v),\ \langle n, v, m, f' \rangle\big) \\
        {}\cup{}& \bigcup_{\mathmakebox[\opsubwd]{((n, v), f) \Delta_a (m, vs, fs)}} D\big((\langle \langle n, f \rangle, g(v) \rangle, v),\ \langle n, v, f,
 m, f' \rangle\big)
      \end{aligned}
    \]
  \end{subdefinition}
\end{definition}

\begin{theorem}[Soundness]
  \label{thm.concurrent-feature-reduction-soundness}
  If $S \in \mathcal{S}(\Delta, r)$ for $(R, \Delta, r)$ produced by Def.~\ref{def.concurrent-feature-reduction}, then
  \[
    \begin{aligned}
      S_{CF} &= \big\{ \big((n, v),\, \{ f \mid (\langle \langle n, f \rangle, g(v) \rangle, v) \in S \}\big) \;\big|\; (\langle n, g(v) \rangle, v) \in S \big\} \\
      \pi &= \{ ((m, u),\, (n, v)) \mid (\langle n, v, m \rangle, u) \in S \}
    \end{aligned}
  \]
  are a valid resolution and parent relation in $\mathcal{S}_{CF}(\Delta_f, \Delta_a, g, r_{CF})$.
\end{theorem}

\begin{theorem}[Completeness]
  \label{thm.concurrent-feature-reduction-completeness}
  If $(S_{CF}, \pi) \in \mathcal{S}_{CF}(\Delta_f, \Delta_a, g, r_{CF})$, define
  \[\mathit{active} = ((n, v),\, \_) \in S_{CF},\ u \in vs,\ fs \subseteq fs',\ ((m, u), fs') \in S_{CF},\ ((m, u),\, (n, v)) \in \pi\]
  a predicate for whether $(n, v)$ selects dependee $(m, u)$, and collect its intermediate packages as
  \[\mathit{witness}(j) = \{ (i, u) \mid \mathit{active} \} \cup \{ (j, u) \mid f' \in fs,\ \mathit{active} \}\]
  then
  \settowidth{\opsubwd}{$\scriptstyle ((n, v), f) \Delta_a (m, vs, fs)$}%
  \[
    \begin{aligned}
      S = {}& \bigcup_{\mathmakebox[\opsubwd]{((n, v),\, fs) \in S_{CF}}} \big( \{ (\langle n, g(v) \rangle, v) \} \cup \{ (\langle \langle n, f \rangle, g(v) \rangle, v) \mid f \in fs \} \big) \\
      {}\cup{}& \bigcup_{\mathmakebox[\opsubwd]{(n, v) \Delta_f (m, vs, fs)}} \mathit{witness}(\langle n, v, m, f' \rangle)
      {}\cup{} \bigcup_{\mathmakebox[\opsubwd]{((n, v), f) \Delta_a (m, vs, fs)}} \mathit{witness}(\langle n, v, f, m, f' \rangle)
    \end{aligned}
  \]
  is a valid resolution in $\mathcal{S}(\Delta, r)$.
\end{theorem}

\begin{figure}[ht]
  \captionsetup{justification=centering}
  \begin{lrbox}{\sfboxa}\begin{minipage}{0.25\textwidth}%
{\scriptsize\bfseries\ttfamily B-1/Cargo.toml}
\begin{Verbatim}[fontsize=\scriptsize]
[dependencies.D]
version = ">=1, <3"
features = ["a"]
\end{Verbatim}
{\scriptsize\bfseries\ttfamily C-1/Cargo.toml}
\begin{Verbatim}[fontsize=\scriptsize]
[dependencies.D]
version = ">=2, <4"
features = ["b"]
\end{Verbatim}
{\scriptsize\bfseries\ttfamily D-1/Cargo.toml}
\begin{Verbatim}[fontsize=\scriptsize]
[features]
a = ["F/c"]
b = ["F/d"]
[dependencies]
F.version = "1"
F.optional = true
\end{Verbatim}
{\scriptsize\bfseries\ttfamily F-1/Cargo.toml}
\begin{Verbatim}[fontsize=\scriptsize]
[features]
c = []
d = []
\end{Verbatim}
  \end{minipage}\end{lrbox}%
  \begin{lrbox}{\sfboxb}\begin{minipage}{0.30\textwidth}\centering
    \[\begin{aligned}
      (A, 1) &\Delta_f (B, \{1\}, \emptyset) \\
      (A, 1) &\Delta_f (C, \{1\}, \emptyset) \\
      (B, 1) &\Delta_f (D, \{1, 2\}, \{\alpha\}) \\
      (C, 1) &\Delta_f (D, \{2, 3\}, \{\beta\}) \\
      ((D, 1), \alpha) &\Delta_a (F, \{1\}, \{\gamma\}) \\
      ((D, 1), \beta) &\Delta_a (F, \{1\}, \{\delta\})
    \end{aligned}\]
  \end{minipage}\end{lrbox}%
  \begin{lrbox}{\sfboxc}\begin{minipage}{0.40\textwidth}\centering
    \includegraphics[scale=0.8]{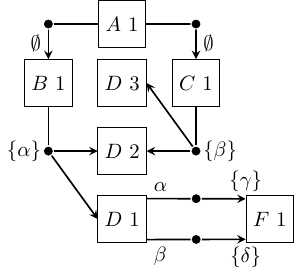}%
  \end{minipage}\end{lrbox}%
  \setlength{\sfrowht}{0pt}\sfmax{\sfboxa}\sfmax{\sfboxb}\sfmax{\sfboxc}%
  \hfill
  \begin{subfigure}[b]{0.25\textwidth}
    \sfcell{\usebox{\sfboxa}}
    \caption{Cargo syntax.}
    \label{fig.concurrent-feature.cargo}
  \end{subfigure}%
  \hfill
  \begin{subfigure}[b]{0.30\textwidth}
    \sfcell{\usebox{\sfboxb}}
    \caption{Concurrent Feature \\ Package Calculus.}
    \label{fig.concurrent-feature.calculus}
  \end{subfigure}%
  \hfill
  \begin{subfigure}[b]{0.40\textwidth}
    \sfcell{\usebox{\sfboxc}}
    \Description{Hypergraph rendering of the concurrent-feature calculus dependencies given alongside.}
    \caption{Hypergraph illustration.}
    \label{fig.concurrent-feature.hypergraph}
  \end{subfigure}%
  \hfill\null
  \caption{Concurrent Feature Package Calculus instance where $g(v) = v$; $D$ supports feature $\alpha$ at versions $1$--$2$ and $\beta$ at versions $1$--$3$, and $F$ supports $\gamma$ and $\delta$ at version $1$.}
  \label{fig.concurrent-feature}
\end{figure}

Fig.~\ref{fig.concurrent-feature} shows a Concurrent Feature Package Calculus instance reduced to the core in Fig.~\ref{fig.concurrent-feature-reduction} (Appendix~\ref{appendix.concurrent-feature-reduction}).

\subsubsection{Composition Considerations}
\label{sec.a-la-carte.composition.considerations}

Reductions satisfy no general composition law: each composition must be defined explicitly, accounting for how the component reductions interact.
The concurrent-feature composition shows the method: defining the semantics of the composed extension and then combining reductions.
We identify three considerations for composing extensions to model real-world package managers.

\paragraph{Invariant Violation.}
\label{sec.a-la-carte.composition.considerations.invariant-violation}
The conflict reduction (Def.~\ref{def.conflict-reduction}) and the negation encoding of the package formula reduction (Def.~\ref{def.package-formula-reduction}) encode mutual exclusion via version uniqueness, which the Concurrent Package Calculus relaxes -- breaking the encoding.
A correct composition must make the reductions mutually aware: either apply the concurrent reduction first and thread its $\langle n, g(v) \rangle$ renaming through, or assign $\epsilon$ granularity to the synthetic packages so they conflict in Def.~\ref{def.concurrent-resolution.version-granularity}.

\paragraph{Interacting Extensions.}
\label{sec.a-la-carte.composition.considerations.interacting-extensions}
Even when invariants are preserved, extensions whose dependency types interact require reductions that are aware of each other's constructs.
For example, Portage combines features~(\S\ref{sec.mise-en-place.features}) with package formulae~(\S\ref{sec.mise-en-place.package-formula}), and a formula atom can carry feature annotations: \texttt{|| ( foo[bar] baz )}.
A correct composition requires the formula atom $(n, vs)$ (Def.~\ref{def.package-formula.grammar}) to support features (e.g.\ $(n, vs, fs)$), and the feature additional dependency target $(n, vs, fs)$ (Def.~\ref{def.feature-dependency.additional}) to support formulae (e.g.\ $\psi$).

\paragraph{Ecosystem Variation.}
\label{sec.a-la-carte.composition.considerations.ecosystem-variation}
The semantics of individual extensions can also differ subtly between ecosystems.
For example, npm supports two different behaviours for peer dependencies depending on the resolver version~(\S\ref{sec.mise-en-place.peer-dependencies}); Cargo features are additive only, while Portage's USE flags can additionally require the absence of a flag~(\S\ref{sec.background.survey.features}) -- effectively a conflict on a feature.
In each case the core calculus is the fixed target; composition reconciles the component reductions without extending the core.

\subsection{Transpiling Packaging Languages}
\label{sec.a-la-carte.transpiling}

Polyglot resolution~(\S\ref{sec.a-la-carte.composition}) resolves dependencies across ecosystems but leaves them expressed in the core.
A more ambitious goal is to \textit{translate} dependencies from one ecosystem's packaging language into another's.
The core makes this concrete by serving as an intermediate representation: rather than $n^2$ translators that each understand a pair of dependency languages, one routes through the core, which reduces the problem to a per-ecosystem frontend and backend -- $2n$ translators -- via the pipeline:
\[L_A \xrightarrow{\,\textup{parse}\,} \Delta_A \xrightarrow{\,\textup{reduce}\,} \Delta \xrightarrow{\,\textup{lift}\,} \Delta_B \xrightarrow{\,\textup{emit}\,} L_B\]
where $L_A$ and $L_B$ denote the source and target DSLs, and $\Delta_A$, $\Delta_B$ the source and target extended-calculus instances.

Casting translation this way isolates where the difficulty lies.
Parsing and emitting are the ecosystem-specific frontend and backend, converting between a DSL and a calculus instance; reduction is as in~\S\ref{sec.mise-en-place}.
The crux is \textit{lifting}: an operation $\textup{lift}$ that reconstructs an extended instance from the core by decoding the synthetic names $\langle \cdot \rangle \in N$ that each reduction introduces.
For a single extension, lifting is a \textit{retraction} of the reduction, $\textup{lift} \circ \textup{reduce} = \mathrm{id}$, recovering the original instance up to normalisation; single-extension translation therefore follows from the reductions of~\S\ref{sec.mise-en-place}.

The open problem is lifting a \textit{composed} reduction.
A real ecosystem combines several extensions, so its target instance $\Delta_B$ must be reconstructed by a composed lifting, and the per-extension liftings do not compose: a composed reduction interleaves the synthetic structures of its constituents -- a concurrent-feature instance nests feature names inside granular names~(Def.~\ref{def.concurrent-feature-reduction}) -- so one cannot lift a single extension out in isolation.
Like a composed reduction, a composed lifting must be defined explicitly, subject to the considerations of~\S\ref{sec.a-la-carte.composition.considerations}.
The calculus therefore reduces cross-ecosystem translation to a single, well-defined problem -- inverting composed reductions -- which we leave, with the \textit{parse} and \textit{emit} stages, to future work.

\section{Related Work}
\label{sec.related}

\paragraph{Dependency Resolution.}
\label{sec.related.dependency-resolution}

The idea of using a formalism for package versioning dates back to the Optimal Package Install/Uninstall Manager (Opium)~\cite{tucker2007opium} in 2007, which adapted Debian's \texttt{apt-get} to use a SAT solver for dependency resolution.
Opium's formalism is close to our Conflict Package Calculus; it encodes Debian's dependencies and conflicts in a SAT problem, expanding version formulae to sets.
We describe how a resolver can prioritise fresher resolutions (Def.~\ref{def.resolution-ordering}); Opium goes further by optimising an objective function over resolutions via pseudo-Boolean and integer linear programming solvers, which it uses to define the install and uninstall problems.

Later work on the Common Upgradeability Description Format (CUDF)~\cite{cudf} generalised the solver interface into a file format that could be consumed by specialised package resolvers.
One such early resolver used answer-set programming (ASP)~\cite{gebser2011asp}, which also found use in high-performance computing (HPC) software infrastructure for solving dependencies with support for a rich set of user-specified resolver preferences~\cite{gamblin2022asp}.
CUDF provides a fixed vocabulary of dependency constructs -- depends, conflicts, and provides -- with integer versions, rich optimisation criteria, and support for upgrade semantics (modelling the currently installed state).
This line of work also proposed modular package-manager architectures that delegate constraint solving to pluggable external solvers~\cite{abate2011mpm}, and developed a formal model of package installation~\cite{dicosmo2013formal} whose conditions -- abundance, peace, and flatness -- correspond to our dependency closure, conflict avoidance, and version uniqueness.
The model notes that the presence of flatness varies between ecosystems -- the divergence which our concurrent-version extension~(\S\ref{sec.mise-en-place.concurrent}) formalises.
However, the fixed vocabulary of CUDF forces package managers to pre-evaluate or expand constructs that do not fit, and CUDF has not seen wide adoption beyond the Debian and opam ecosystems.

Abate et al.~\cite{abate2020dependency} observed in 2020 that while SAT-based dependency solving did find traction among some package manager implementations, reusable components for version solving had not seen wider adoption.
The Ecosyste.ms project~\cite{ecosystems-resolvers} catalogues resolver algorithms across ecosystems, classifying resolution strategies for over thirty package managers.
In recent years, some of the larger package ecosystems, such as npm, have come under severe scaling pressure, and it is becoming increasingly important to revisit the question of how to reuse knowledge across ecosystems.
For example, PacSolve's MaxSMT-based resolver is a drop-in replacement for npm that reduced vulnerabilities in resolved packages and picked fresher and fewer dependees~\cite{pinckney2023pacsolve,pinckney2024phd}.
Approximate solving has also been useful for package selection in bounded-latency situations~\cite{ignatiev2014optimization}.
However, these systems are all tied to a single ecosystem: none can resolve dependencies that cross ecosystem boundaries.
PubGrub~\cite{weizenbaum2018pubgrub} is a CDCL-based dependency resolution algorithm created for Pub (Dart's package manager) that has since been adopted by Bundler (Ruby), Poetry (Python), and Swift Package Manager, among others.
The pubgrub-rs Rust implementation~\cite{pubgrub-rs} encodes concurrent versions and features with the same techniques we formalise in~\S\ref{sec.mise-en-place.concurrent} and~\S\ref{sec.mise-en-place.features}.

More abstract treatments model dependency metadata itself as a mathematical object -- a general event structure~\cite{bazerman2021events}, or a \textit{dependency structure with choice} isomorphic to an antimatroid~\cite{bazerman2024mathematical} -- characterising the space of dependency-closed states rather than the resolution problem we address.

\paragraph{Build Systems.}
\label{sec.related.build-systems}
Package management is closely tied to build systems, which are also surprisingly poorly specified~\cite{mokhov2018build}.
They are sometimes unified in newer toolchains (Table~\ref{tbl.comparison}), but doing so makes it difficult to compose projects that span ecosystems (consider the difficulty of depending on a Rust library from OCaml code with just two package managers).
The Package Calculus specifies dependency resolution as a formalism independent of the underlying build systems, preserving across ecosystems the separation that unified toolchains lose.
In the future, creating a composable formal theory of build systems and package managers~\cite{agnarsson1985} would let us combine cross-ecosystem codebases to conduct correctness testing across the vast amount of source code published~\cite{macho2024dvalidator}.

\paragraph{Software Supply Chain Security.}
\label{sec.related.supply-chain-security}
The increasing complexity of cross-ecosystem dependency chains puts pressure on the security of software supply chains, both in the choice of dependees selected for a software project~\cite{alfadel2023python,alfadel2023npm,zhang2023maven} and in vulnerabilities in the many package managers available~\cite{cappos2008attacks,bos2023attacks}.
Software Bills of Materials (SBOMs) are being rapidly adopted~\cite{cofano2024sbom} to specify the full set of versioned dependees that go into a given application.
However, the same problem of uneven quality across ecosystems applies here, with the JavaScript ecosystem being particularly unpredictable~\cite{rabbi2024sbom}.
The Package Calculus could be used to generate SBOMs systematically, reduce transitive trust dependencies across language and OS ecosystems~\cite{schwaighofer2024extending}, drop resource usage by debloating unnecessary packages~\cite{pashakhanloo2022pacjam}, and reduce the latency of applying security updates in third-party dependees~\cite{rahkema2022dependencies,stringer2020lag}.

\section{Conclusion}
\label{sec.conclusion}

We surveyed package managers, identifying a shared core and the axes along which they diverge~(\S\ref{sec.background}).
We formalised this core as the Package Calculus~(\S\ref{sec.calculus}).
We modelled each axis of divergence as an extension with a reduction back to the core~(\S\ref{sec.mise-en-place}).
We showed how extensions compose to model package managers that combine several of them, providing the foundation for polyglot resolution across ecosystems, and discussed how the calculus might further serve as an intermediate representation for translation between ecosystems, reducing $n^2$ translators to $2n$~(\S\ref{sec.a-la-carte}).
Building on this foundation, future work can develop principled cross-ecosystem tooling~(\S\ref{sec.introduction}) via translators between ecosystem DSLs, polyglot resolvers, and ecosystem-specific resolvers verified against the calculus's specifications.

\begin{acks}
  We thank Raito Bezarius, Michael Winston Dales, Kate Deplaix, Roberto Di Cosmo, Mark Elvers, Arjun Guha, Sadiq Jaffer, Ranjit Jhala, Artsiom Karakin, Shriram Krishnamurthi, Thomas Leonard, Jon Ludlam, Andrew Nesbitt, Simon Peyton Jones, Swapnil Raj, James Rowbottom, Jacob Trevor, and Garrett Wollman for their contributions and comments.
  Simon Peyton Jones in particular spent much time giving us layout and notation suggestions that greatly improved the presentation.
  We also thank the anonymous ICFP reviewers for their feedback.
\end{acks}

\section*{Funding}

This research was supported by unrestricted donations from the Huawei HiSilicon Scholarship, Tarides, and Jane Street.

\section*{Data-Availability Statement}

The Lean~4 mechanisation of the theorems in this paper is available as an archive on Zenodo~\cite{artifact}.
The latest version is maintained at \url{https://github.com/RyanGibb/package-calculus/}.

\appendix

\section{A Short History of Package Managers}
\label{appendix.history}

The original Unix and pre-Unix operating systems had no concept of automated package management; the operating system was distributed with libraries and utilities, and third-party software had to be built and installed manually.
System V Release 4.0, announced in 1988, included the `packaging tool' \texttt{pkgadd} and related commands to automate the process of installing packages onto the system from a source such as a tape, CD-ROM, or networked filesystem, as well as updating and removing them~\cite{svr4guide}.
It verified that dependees were installed, but did not automatically select and install compatible versions.
Other commercial Unix vendors created similar systems in the following years, as did early Linux (1991) distributions: Slackware's \texttt{pkgtool} (1993), Debian's \texttt{dpkg} (1994), and Red Hat's \texttt{rpm} (1995).
FreeBSD's Ports (1994) introduced automatic dependee installation, relying on regular expressions on package names to ensure application binary interface (ABI) compatibility.

Early programming-language package managers emerged in the form of archive networks for distributing software over the Internet with the Comprehensive \TeX\ Archive Network (CTAN) in 1992~\cite{ctan}, the Comprehensive Perl Archive Network (CPAN) in 1995~\cite{cpan}, and the Comprehensive R Archive Network (CRAN) in 1997~\cite{cran}.
The Perl module \texttt{CPAN.pm} and CRAN's \texttt{install.packages()} allowed for the automatic installation of dependees, but always selected the latest indexed version of each, performing no version solving.
CTAN had no native automated installer; package management was later provided by \TeX\ Live's \texttt{tlmgr} (2008), which likewise installed only the latest versions, though \TeX\ Live mitigated incoherence by distributing packages as a curated annual snapshot rather than a continuously updated live index.

Debian's Advanced Packaging Tool (APT) (1998), building atop the existing \texttt{dpkg}, resolved compatible versions of dependencies based on version constraints and automatically installed them.
The Yellowdog Updater, Modified (YUM) (2002), building on \texttt{rpm}, brought similar capabilities to Red Hat and other RPM-based distributions, and was later superseded by Dandified YUM (DNF) (2013).
Dependency resolution in these systems has been shown to be NP-complete~\cite{dicosmo2005edos}, including for upgrades~\cite{abate2012}.
Slackware is a notable exception to the norm in dependency resolution; official mechanisms do not provide any support for dependency tracking, instead relying on the user to manually install dependees of a package.
At the other extreme, Nix~(2003)~\cite{dolstra2004nix} opted out of version solving by requiring exact versions of dependees, deploying packages at paths containing cryptographic hashes~(Appendix~\ref{appendix.singular-dependencies}).

By the mid-2000s, automatic dependency resolution with version constraints had become standard practice.
Package managers began to be built as monolithic systems expected of any new operating system or programming language.
This evolutionary pattern extends beyond operating systems and programming languages: plugin ecosystems for web browsers, text editors, and applications have similarly progressed through centralised archives, automatic installation, and dependency resolution.
More recently, tools originally designed for CI, chart templating, and infrastructure-as-code -- GitHub Actions, Helm, and Terraform -- have grown transitive dependency trees and inherited the problems of package management whether or not they were designed for it~\cite{nesbitt2026quacks}.
As these package managers have developed, the semantics of their dependencies have diverged~(\S\ref{sec.background.survey}), each introducing functionality to address its ecosystem's specific needs.

\section{Thm.~\ref{thm.resolution-complexity}: \textsc{DependencyResolution} Complexity}
\label{appendix.resolution-complexity}

\begin{proof}
  We establish both membership and hardness:

  \textbf{Membership in NP}: a candidate resolution $S$ can be verified in polynomial time by checking
  root inclusion (Def.~\ref{def.calculus.resolution.root-inclusion}) in $O(|S|)$ time,
  dependency closure (Def.~\ref{def.calculus.resolution.dependency-closure}) in $O(|S|\cdot|\Delta|)$ time,
  and version uniqueness (Def.~\ref{def.calculus.resolution.version-uniqueness}) in $O(|S|^{2})$ time.

  \textbf{NP-Hardness}:
  we give a polynomial-time reduction from \textsc{3-SAT}, which is well-known to be NP-complete~\cite{cook1971sat}.
  Given a 3-CNF formula with variables $x_1, \dots, x_n$ and clauses $c_1, \dots, c_m$ (each $c_j = l_1 \lor l_2 \lor l_3$), we build a resolution instance with root $r = (q, \epsilon)$ and packages
  \[
    R = \{(q,\epsilon)\}
      \cup \bigcup_i \{(x_i,\top),(x_i,\bot)\}
      \cup \bigcup_j \{(c_j,l_1),(c_j,l_2),(c_j,l_3)\},
  \]
  whose dependencies are, for every clause $c_j$,
  \[
    (q,\epsilon)\,\Delta\,(c_j,\{l_1,l_2,l_3\})
    \qquad\text{and, for each } k \in \{1,2,3\},\qquad
    (c_j,l_k)\,\Delta\,(\text{var}(l_k),\{\text{pol}(l_k)\}),
  \]
  where $\text{var}(l)$ is the variable of literal $l$, and $\text{pol}(l)$ is $\top$ for positive literals and $\bot$ for negative.

  The construction runs in $O(n+m)$ time.
  We prove the equivalence:

  \noindent$(\Rightarrow)$ Given a satisfying assignment $\sigma$, for each clause $c_j = l_1 \lor l_2 \lor l_3$ let $k_j = \min\{k \in \{1,2,3\} \mid \sigma(l_k) = \top\}$.
  Then $S = \{r\} \cup \{(c_j, l_{k_j}) \mid 1 \leq j \leq m\} \cup \{(x_i, \sigma(x_i)) \mid 1 \leq i \leq n\}$ satisfies all dependencies by construction and maintains version uniqueness.

  \noindent$(\Leftarrow)$ Any resolution $S \in \mathcal{S}(\Delta, r)$ induces a satisfying assignment $\sigma(x_i) = v$ where $(x_i, v) \in S$ (arbitrary otherwise), since each clause package $c_j$ requires at least one of its literals to be satisfied.

  As \textsc{DependencyResolution} is in NP and is NP-hard, we conclude that it is NP-complete.
\end{proof}

\section{SAT-Based Dependency Resolution}
\label{appendix.sat-resolution}

Modern package managers~\cite{cudf} often employ SAT-based resolution to solve dependency problems.
We encode \textsc{DependencyResolution} as \textsc{SAT} as follows.

\begin{definition}[Package Calculus SAT Encoding]
  \label{def.package-calculus-sat-encoding}
  Given real packages $R$, dependency relation $\Delta$, and root $r$, where $X_p$ denotes inclusion of package $p$, we define an encoding to the SAT instance $\Sigma \coloneqq (X_r) \land \Sigma_\Delta \land \Sigma_{\neq}$ where:
  \[\Sigma_\Delta \coloneqq \bigwedge_{p \Delta (n, vs)} \left(\neg X_p \lor \bigvee_{\substack{v \in vs \\ (n, v) \in R}} X_{(n, v)}\right)
  \qquad\qquad
  \Sigma_\neq \coloneqq \bigwedge_{\substack{(n, v), (n, v') \in R \\ v <_v v'}} (\neg X_{(n, v)} \lor \neg X_{(n, v')})\]
  Note that for a package name $n \in N$ with $k = |\{v \mid (n, v) \in R\}|$ versions, this adds $O(k^2)$ clauses, but a SAT solver with an at-most-one primitive can reduce this to $O(k)$~\cite{nguyen2021amo}.
\end{definition}

\begin{theorem}[Soundness]
  \label{thm.package-calculus-sat-encoding-soundness}
  For any assignment $\sigma : \{ X_p \mid p \in N \times V \} \to \{\top, \bot\}$ that satisfies $\Sigma$, the set $S$ defined by
  $\{p \in R \mid \sigma(X_p) = \top\}$
  is a valid resolution in $\mathcal{S}(\Delta, r)$.
\end{theorem}

\begin{theorem}[Completeness]
  \label{thm.package-calculus-sat-encoding-completeness}
  For any $S \in \mathcal{S}(\Delta, r)$, the assignment
  $\sigma(X_p) = \top \iff p \in S$
  satisfies $\Sigma$.
\end{theorem}

\begin{example}
  \label{ex.sat-encoding}
  Consider the dependencies from Fig.~\ref{fig.calculus} with query $Q = \{(A, \{1\})\}$.
  The SAT encoding produces $\Sigma = X_r \land \Sigma_\Delta \land \Sigma_\neq$ where:
  \small
  \begin{align*}
      \Sigma_\Delta & \coloneqq (\neg X_r \lor X_{A1}) \land (\neg X_{A1} \lor X_{B1}) \land (\neg X_{A1} \lor X_{C1}) \land (\neg X_{B1} \lor X_{D1} \lor X_{D2}) \land (\neg X_{C1} \lor X_{D2} \lor X_{D3}) \\
      \Sigma_\neq & \coloneqq (\neg X_{D1} \lor \neg X_{D2}) \land (\neg X_{D1} \lor \neg X_{D3}) \land (\neg X_{D2} \lor \neg X_{D3})
  \end{align*}
\end{example}

\begin{definition}[Resolution Ordering]
  \label{def.resolution-ordering}
  A \textit{resolution ordering} is a partial order $\leq_{\mathcal{S}}$ on $\mathcal{S}(\Delta, r)$ where for $S_1, S_2 \in \mathcal{S}(\Delta, r)$, we say
  $S_1 \leq_{\mathcal{S}} S_2 \iff \forall\, (n, v_2) \in S_2.\, \exists\, v_1 \in V.\, (n, v_1) \in S_1 \land v_1 \leq_v v_2$.
  Every package name in $S_2$ has a version at least as great as that in $S_1$ according to $\leq_v$ (Def.~\ref{def.version-ordering}).
\end{definition}
Maximal elements of $\leq_{\mathcal{S}}$ represent the `freshest' possible resolutions, but a maximum resolution may not exist.
For example, given $(A, 1) \Delta (B, \{1, 2\})$, $(A, 1) \Delta (C, \{1, 2\})$, and $(B, 2) \Delta (C, \{1\})$, the query $Q = \{(A, \{1\})\}$ has maximal resolutions $\{r, (A, 1), (B, 1), (C, 2)\}$ and $\{r, (A, 1), (B, 2), (C, 1)\}$ but no maximum resolution.
Finding a valid resolution (Def.~\ref{def.calculus.resolution}) does not depend on version ordering, but an algorithm can prioritise higher versions (Def.~\ref{def.ordered-sat-encoding}).

\begin{definition}[Ordered SAT Encoding]
  \label{def.ordered-sat-encoding}
  Given the Package Calculus SAT encoding (Def.~\ref{def.package-calculus-sat-encoding}), we refine the dependency implication clauses using version ordering $\leq_v$,
  \[\Sigma_\Delta \coloneqq \bigwedge_{p \Delta (n, vs)} (\neg X_p \lor X_{(n, v_1)} \lor \cdots \lor X_{(n, v_k)}) \text{ where } \{v_1, \ldots, v_k\} = \{v \in vs \mid (n, v) \in R\},\ v_1 \geq_v \cdots \geq_v v_k\]
  This encoding ensures that in each dependency clause, higher versions appear first.
  SAT solvers prioritising leftmost literals will preferentially assign true to earlier (higher-version) literals.
  The ordering transforms version preference into a search heuristic without altering satisfiability.
\end{definition}

\section{Singular Dependencies}
\label{appendix.singular-dependencies}

Singular dependencies only restrict the core calculus, rather than extending it.
Nix, Guix~(Various)~\cite{courtes2013functional}, Unison, and Lake~(Lean~4) support only a single, exact version of a package name as a dependee, delegating the responsibility of version solving from the package manager to the packager.
Where the core calculus expresses a dependency on a \textit{set} of compatible versions and leaves it to the resolver to select among them, singular dependencies require the packager to commit to a specific version at packaging time.
The reduction to the core is trivial, but because singular dependencies are \textit{less} expressive than the core, we cannot reduce an instance of the core calculus to singular dependencies.

\begin{definition}[Singular Dependency]
  \label{def.singular-dependency}
  We define singular dependencies $\beta \subseteq (N \times V) \times (N \times V)$ as a relation where an element $p \beta d$ denotes a package $p$ expressing a dependency on a package $d$.
\end{definition}

\begin{definition}[Singular Dependency Resolution]
  \label{def.singular-dependency-resolution}
  Given singular dependencies $\beta$ and root $r_\beta$, a resolution $S_\beta \subseteq R_\beta$ is valid if:
  \begin{subdefinition}
    \item\label{def.singular-dependency-resolution.root-inclusion}
    \textbf{Root Inclusion}:
    $r_\beta \in S_\beta$, as in Def.~\ref{def.calculus.resolution.root-inclusion}.
    \item\label{def.singular-dependency-resolution.dependency-closure}
    \textbf{Dependency Closure}:
    $\forall\, p \in S_\beta.\, p \beta d \implies d \in S_\beta$
    \item\label{def.singular-dependency-resolution.version-uniqueness}
    \textbf{Version Uniqueness}:
    $\forall\, (n, v),\, (n, v') \in S_\beta.\, v = v'$, as in Def.~\ref{def.calculus.resolution.version-uniqueness}.
  \end{subdefinition}
\end{definition}

Nix DSL expressions contain conditional logic and feature-like parameters as function arguments when creating derivations.
Nix could add support for resolving version constraints as the Package Calculus does, either in the DSL or in the store derivation format, but the semantics of this would need to be defined.

For example, if we can say our Nix derivation depends on a package name with version constraints, can we parameterise this by a feature?
Since this feature selection can pull in different dependencies depending on the version of the package selected, this has to be part of the dependency resolution rather than simply a function parameter in the Nix DSL.
Even with version-constraint resolution added, such a Nix would need a mechanism to restrict this for ecosystems where it is not possible, for example, when libraries need to link together in a binary.

This explains the proliferation of tools such as \texttt{cargo2nix}, \texttt{opam-nix}, and \texttt{poetry2nix}, doing dependency resolution in ecosystem-specific tooling and locking versions in Nix derivations without any way to express version constraints across ecosystems.
In Table~\ref{tbl.comparison}, Nix is not listed as supporting most extensions, as it does not support them as part of version-constrained dependency resolution.

\section{Optional Dependencies in Build Graphs}
\label{appendix.optional-dependencies}

We promised in~\S\ref{sec.background.survey.optional-dependencies} that optional dependencies are a build-ordering concern, not a resolution concern, which we now make precise.
A valid resolution (Def.~\ref{def.calculus.resolution}) determines \textit{which} packages to install, but not the \textit{order} in which to install them.
Package managers like opam, as part of their deployment~(\S\ref{sec.background.common-core.deployment}), must linearise the resolution into a sequence of build actions such that all the dependees of a package are built before the depender itself.
This requires a topological sort of the dependency graph induced by the resolution -- which is only possible when that graph is acyclic.

In practice, dependency graphs can contain cycles.
In Debian, for example, \texttt{libgcc-s1} depends on \texttt{libc6} (it requires the C~library at run time), while \texttt{libc6} depends on \texttt{libgcc-s1} (it requires the GCC run-time support library).
Both packages are valid in a resolution, but no linear installation order can satisfy both ordering constraints simultaneously.
APT handles this by breaking cycles at an arbitrary point and relying on \texttt{dpkg}'s ability to partially unpack packages before their dependees are fully configured~\cite{debian}.
In opam, packages can be marked with the \texttt{post} variable, which indicates that the dependee is not required at build time, allowing the cycle to be broken explicitly.
Other package managers avoid the problem entirely: npm lays out files in nested directories and relies on Node.js's module loader to handle cycles between modules.

\begin{definition}[Build Graph]
  \label{def.build-graph}
  We formalise this installation ordering as the \textit{build graph}.
  Given dependencies $\Delta$ (Def.~\ref{def.calculus.dependency}) and a resolution $S$ (Def.~\ref{def.calculus.resolution}), the build graph is a directed graph with vertices $S$ and edges $E = \{ (p, (n, v)) \mid p \in S,\ p \Delta (n, vs),\ (n, v) \in S \}$.
\end{definition}

\begin{definition}[Optional Dependency]
  \label{def.optional-dependency}
  We define optional dependencies as a relation $O \subseteq (N \times V) \times (N \times \mathcal{P}(V))$.
  We denote an element of $O$ as $p O (n, vs)$ where package $p$ will use package name $n$ with version $v \in vs$ if $(n, v)$ is in the resolution.
\end{definition}

Optional dependees are required before their depender in the build graph (Def.~\ref{def.build-graph}).
\[E = \{ (p, (n, v)) \mid p \in S,\ p \Delta (n, vs),\ (n, v) \in S \} \cup \{ (p, (n, v)) \mid p \in S,\ p O (n, vs),\ v \in vs,\ (n, v) \in S \}\]

\begin{example}
  \label{ex.optional-dependency}
  Consider,
  \[(A, 1) \Delta (B, \{1\}) \qquad (A, 1) \Delta (D, \{1\}) \qquad (B, 1) \Delta (C, \{1\}) \qquad (B, 1) O (D, \{1\})\]
  If $Q = \{(B, \{1\})\}$, then $\{r, (B, 1), (C, 1)\} \in \mathcal{S}(\Delta, r)$, and its build graph is $\{(r, (B, 1)), ((B, 1), (C, 1))\}$.
  If $Q = \{(A, \{1\})\}$, then $\{r, (A, 1), (B, 1), (C, 1), (D, 1)\} \in \mathcal{S}(\Delta, r)$ and $E \ni ((B, 1), (D, 1))$.
\end{example}

We can trivially reduce optional dependencies to the core calculus by omitting them; they have no effect on the set of resolutions $\mathcal{S}(\Delta, r)$.
Package $B$ can be resolved regardless of $D$'s presence.
When assembling the build graph, if $p O (n, vs)$ and $(n, v) \in S$ with $v \in vs$, then $(p, (n, v)) \in E$ \textit{always} and there is no mechanism to disable this edge.
This model captures opam's behaviour where optional dependencies are purely a build-ordering mechanism with no resolution-time effects.

\section{Concurrent Feature Package Calculus Reduction}
\label{appendix.concurrent-feature-reduction}

\begin{figure}[ht]
  \captionsetup{justification=centering}
  \begin{subfigure}[b]{\textwidth}
    \hfill
    \begin{minipage}[t]{0.40\linewidth}
      \footnotesize
      \begin{align*}
        (\langle A, 1 \rangle, 1) &\Delta (\langle A, 1, B \rangle, \{1\}) \\
        (\langle A, 1, B \rangle, 1) &\Delta (\langle B, 1 \rangle, \{1\}) \\
        (\langle A, 1 \rangle, 1) &\Delta (\langle A, 1, C \rangle, \{1\}) \\
        (\langle A, 1, C \rangle, 1) &\Delta (\langle C, 1 \rangle, \{1\}) \\
        (\langle B, 1 \rangle, 1) &\Delta (\langle B, 1, D \rangle, \{1, 2\}) \\
        (\langle B, 1 \rangle, 1) &\Delta (\langle B, 1, D, \alpha \rangle, \{1, 2\}) \\
        (\langle C, 1 \rangle, 1) &\Delta (\langle C, 1, D \rangle, \{2, 3\}) \\
        (\langle C, 1 \rangle, 1) &\Delta (\langle C, 1, D, \beta \rangle, \{2, 3\}) \\
        (\langle B, 1, D \rangle, 1) &\Delta (\langle D, 1 \rangle, \{1\}) \\
        (\langle B, 1, D \rangle, 2) &\Delta (\langle D, 2 \rangle, \{2\}) \\
        (\langle B, 1, D, \alpha \rangle, 1) &\Delta (\langle \langle D, \alpha \rangle, 1 \rangle, \{1\}) \\
        (\langle B, 1, D, \alpha \rangle, 2) &\Delta (\langle \langle D, \alpha \rangle, 2 \rangle, \{2\}) \\
        (\langle B, 1, D, \alpha \rangle, 1) &\Delta (\langle B, 1, D \rangle, \{1\}) \\
        (\langle B, 1, D, \alpha \rangle, 2) &\Delta (\langle B, 1, D \rangle, \{2\}) \\
        (\langle C, 1, D \rangle, 2) &\Delta (\langle D, 2 \rangle, \{2\}) \\
        (\langle C, 1, D \rangle, 3) &\Delta (\langle D, 3 \rangle, \{3\}) \\
        (\langle C, 1, D, \beta \rangle, 2) &\Delta (\langle \langle D, \beta \rangle, 2 \rangle, \{2\}) \\
        (\langle C, 1, D, \beta \rangle, 3) &\Delta (\langle \langle D, \beta \rangle, 3 \rangle, \{3\})
      \end{align*}
    \end{minipage}%
    \hfill
    \begin{minipage}[t]{0.40\linewidth}
      \footnotesize
      \begin{align*}
        (\langle C, 1, D, \beta \rangle, 2) &\Delta (\langle C, 1, D \rangle, \{2\}) \\
        (\langle C, 1, D, \beta \rangle, 3) &\Delta (\langle C, 1, D \rangle, \{3\}) \\
        (\langle \langle D, \alpha \rangle, 1 \rangle, 1) &\Delta (\langle D, 1 \rangle, \{1\}) \\
        (\langle \langle D, \alpha \rangle, 2 \rangle, 2) &\Delta (\langle D, 2 \rangle, \{2\}) \\
        (\langle \langle D, \beta \rangle, 1 \rangle, 1) &\Delta (\langle D, 1 \rangle, \{1\}) \\
        (\langle \langle D, \beta \rangle, 2 \rangle, 2) &\Delta (\langle D, 2 \rangle, \{2\}) \\
        (\langle \langle D, \beta \rangle, 3 \rangle, 3) &\Delta (\langle D, 3 \rangle, \{3\}) \\
        (\langle \langle D, \alpha \rangle, 1 \rangle, 1) &\Delta (\langle D, 1, F \rangle, \{1\}) \\
        (\langle \langle D, \beta \rangle, 1 \rangle, 1) &\Delta (\langle D, 1, F \rangle, \{1\}) \\
        (\langle D, 1, F \rangle, 1) &\Delta (\langle F, 1 \rangle, \{1\}) \\
        (\langle \langle D, \alpha \rangle, 1 \rangle, 1) &\Delta (\langle D, 1, \alpha, F, \gamma \rangle, \{1\}) \\
        (\langle D, 1, \alpha, F, \gamma \rangle, 1) &\Delta (\langle \langle F, \gamma \rangle, 1 \rangle, \{1\}) \\
        (\langle D, 1, \alpha, F, \gamma \rangle, 1) &\Delta (\langle D, 1, F \rangle, \{1\}) \\
        (\langle \langle D, \beta \rangle, 1 \rangle, 1) &\Delta (\langle D, 1, \beta, F, \delta \rangle, \{1\}) \\
        (\langle D, 1, \beta, F, \delta \rangle, 1) &\Delta (\langle \langle F, \delta \rangle, 1 \rangle, \{1\}) \\
        (\langle D, 1, \beta, F, \delta \rangle, 1) &\Delta (\langle D, 1, F \rangle, \{1\}) \\
        (\langle \langle F, \gamma \rangle, 1 \rangle, 1) &\Delta (\langle F, 1 \rangle, \{1\}) \\
        (\langle \langle F, \delta \rangle, 1 \rangle, 1) &\Delta (\langle F, 1 \rangle, \{1\})
      \end{align*}
    \end{minipage}
    \hfill\null
    \caption{Reduction to the core calculus.}
    \label{fig.concurrent-feature-reduction.reduction}
  \end{subfigure}%
  \vspace{1em}
  \begin{subfigure}[b]{\textwidth}
    \centering
    \includegraphics[scale=0.8]{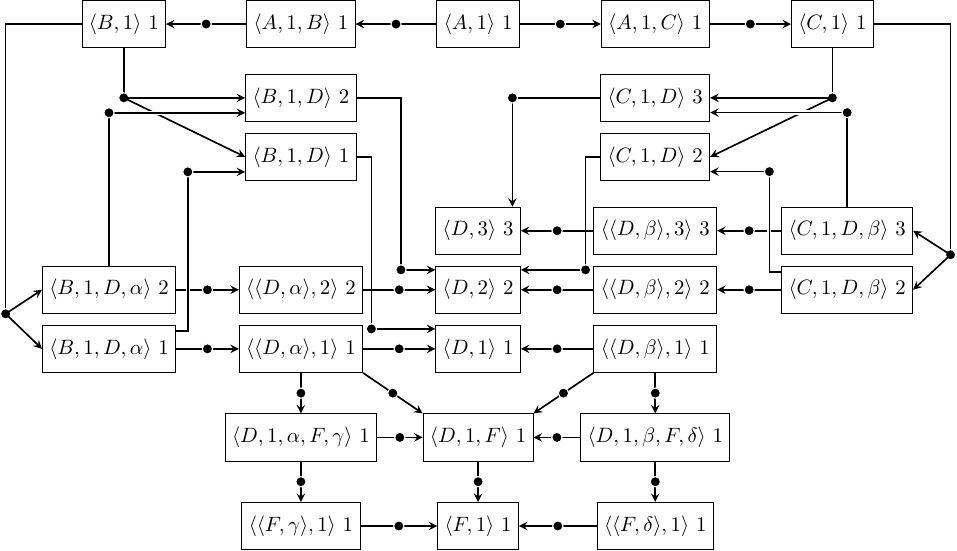}
    \Description{Hypergraph rendering of the reduction formula given alongside.}
    \caption{Hypergraph of reduction.}
    \label{fig.concurrent-feature-reduction.hypergraph}
  \end{subfigure}%
  \caption{Reduction of Fig.~\ref{fig.concurrent-feature} to the core calculus.}
  \label{fig.concurrent-feature-reduction}
\end{figure}

\bibliography{reference}

%% file: table.tex
\newcommand\chk{\color{black}{\ding{51}}}
\newcommand\crs{}%
\newcommand\naa{-}
\newcounter{tabfn}
\newcommand{\tfntext}[2]{%
  \refstepcounter{tabfn}%
  \noindent\textsuperscript{\arabic{tabfn}}%
  \raisebox{7pt}[0pt][0pt]{\phantomsection\label{tabfn-#1}}%
  #2\par
}
\newcommand{\tfnref}[1]{\hyperref[tabfn-#1]{\textsuperscript{\ref{tabfn-#1}}}}
\newcommand{\chkf}[1]{\chk\rlap{\tfnref{#1}}}

\afterpage{

\let\oldfootnotesize\footnotesize
\renewcommand{\footnotesize}{\fontsize{7pt}{0pt}\selectfont}

\thispagestyle{empty}
\setlength{\abovecaptionskip}{4pt}
\setlength{\belowcaptionskip}{0pt}
\setlength{\multicolsep}{0pt}
\begin{landscape}
\begin{table}
\fontsize{6.25pt}{9pt}\selectfont
\setlength{\tabcolsep}{4.8pt}
\captionsetup{justification=centering}
\caption{A comparison of representative package managers.}
\label{tbl.comparison}
\rowcolors{3}{gray!10}{white}
\begin{tabular}{lclll|ccccccccc}
\multicolumn{5}{c|}{Description} & \multicolumn{9}{c}{Dependency Expression}\\
\bf \makecell[b]{Package\\Manager}
& \bf Year
& \bf \makecell[b]{Ecosystem\\\S\ref{sec.background.survey.ecosystem}}
& \bf \makecell[b]{Toolchain\tfnref{toolchain-legend}\\\S\ref{sec.background.survey.toolchain}}
& \bf \makecell[b]{Packaging\\Language\\\S\ref{sec.background.survey.language}}
& \bf \makecell[b]{Version\\Constraints\\\S\ref{sec.background.survey.version-constraints}}
& \bf \makecell[b]{Conflicts\\\S\ref{sec.background.survey.conflicts}}
& \bf \makecell[b]{Concurrent\\Versions\\\S\ref{sec.background.survey.concurrent}}
& \bf \makecell[b]{Peer\\Deps\\\S\ref{sec.background.survey.peer}}
& \bf \makecell[b]{Features\\\S\ref{sec.background.survey.features}}
& \bf \makecell[b]{Package\\Formulae\\\S\ref{sec.background.survey.package-formula}}
& \bf \makecell[b]{Variable\\Formulae\\\S\ref{sec.background.survey.variable-formula}}
& \bf \makecell[b]{Virtual\\Packages\\\S\ref{sec.background.survey.virtual-packages}}
& \bf \makecell[b]{Optional\\Deps\\\S\ref{sec.background.survey.optional-dependencies}}
\\
\hline
CPAN                                & 1995 & Perl                & \bcircle{P}                               & \texttt{cpanfile}               & \chkf{cpan-greedy}
                                                                                                                                                      & \crs & \crs & \crs & \crs & \crs & \crs & \crs & \crs \\
\texttt{pkg\_add}                   & 1996 & OpenBSD             & \bcircle{P}                               & BSD Make                        & \chk & \chk & \crs & \crs & \crs & \chk & \crs & \crs & \crs \\
CRAN                                & 1997 & R                   & \bcircle{P}                               & \texttt{DESCRIPTION}            & \chk & \crs & \crs & \crs & \crs & \crs & \crs & \crs & \crs \\
APT                                 & 1998 & Debian Linux        & \bcircle{P}                               & \texttt{.deb} control           & \chk & \chk & \crs & \crs & \crs & \chk & \chk & \chk & \crs \\
Portage                             & 2000 & Gentoo Linux        & \bcircle{P}                               & ebuild (Bash)                   & \chk & \chk & \chkf{portage-slotting}
                                                                                                                                                                    & \crs & \chk & \chk & \chk & \chk & \crs \\
pacman                              & 2002 & Arch Linux          & \bcircle{P}                               & \texttt{PKGBUILD} (Bash)        & \chk & \chk & \crs & \crs & \crs & \crs & \crs & \chk & \crs \\
Nix~\cite{dolstra2004nix}\tfnref{nix-singular}
                                    & 2003 & Various             & \bcircle{P} \bcircle{B}                   & Nix expressions / \texttt{.drv} & \crs & \crs & \chk & \crs & \crs & \crs & \crs & \crs & \crs \\
Zero Install                        & 2003 & Various             & \bcircle{P}                               & \texttt{feed.xml}               & \chk & \crs & \chk & \crs & \crs & \crs & \chk & \crs & \crs \\
Maven
                                    & 2004 & Java                & \bcircle{P} \bcircle{B}                   & \texttt{pom.xml}                & \chkf{maven-nearest-wins}
                                                                                                                                                      & \crs & \crs & \crs & \crs & \crs & \chk & \crs & \crs \\
Gem                                 & 2004 & Ruby                & \bcircle{P}                               & Ruby                            & \chk & \crs & \crs & \crs & \crs & \crs & \crs & \crs & \crs \\
Cabal~\cite{haskell}                & 2005 & Haskell             & \bcircle{P} \bcircle{B}                   & \texttt{.cabal}                 & \chk & \crs & \chk & \crs & \chk & \crs & \chk & \crs & \crs \\
APK                                 & 2005 & Alpine Linux        & \bcircle{P}                               & \texttt{APKBUILD} (shell)       & \chk & \chk & \crs & \crs & \crs & \crs & \crs & \chk & \crs \\
tlmgr                               & 2008 & \TeX\ Live          & \bcircle{P}                               & \texttt{.tlpsrc}                & \crs & \crs & \crs & \crs & \crs & \crs & \crs & \crs & \crs \\
pip\tfnref{pip-poetry-uv}
                                    & 2008 & Python              & \bcircle{P}                               & \texttt{pyproject.toml}         & \chk & \crs & \crs & \crs & \chk & \crs & \chk & \crs & \crs \\
Homebrew                            & 2009 & macOS               & \bcircle{P}                               & Ruby                            & \crs & \chk & \chkf{homebrew-manual}
                                                                                                                                                                    & \crs & \crs & \crs & \chk & \crs & \crs \\
npm                                 & 2010 & JavaScript          & \bcircle{P}                               & \texttt{package.json}           & \chk & \crs & \chk & \chk & \crs & \crs & \chk & \crs & \crs \\
NuGet                               & 2010 & .NET                & \bcircle{P} \bcircle{B}                   & \texttt{.csproj}                & \chkf{nuget-lowest}
                                                                                                                                                      & \crs & \crs & \crs & \crs & \crs & \chk & \crs & \crs \\
Chocolatey                          & 2011 & Windows             & \bcircle{P}                               & \texttt{.nuspec}                & \chkf{chocolatey-nuget}
                                                                                                                                                      & \crs & \crs & \crs & \crs & \crs & \crs & \crs & \crs \\
Composer                            & 2012 & PHP                 & \bcircle{P}                               & \texttt{composer.json}          & \chk & \chk & \crs & \crs & \crs & \crs & \crs & \chk & \crs \\
Conda                               & 2012 & Python              & \bcircle{P}                               & \texttt{meta.yaml}              & \chk & \chk & \crs & \crs & \crs & \crs & \chk & \crs & \crs \\
\texttt{pkg}\tfnref{freebsd-pkg-ng} & 2012 & FreeBSD             & \bcircle{P}                               & BSD Make                        & \crs & \chk & \crs & \crs & \crs & \crs & \crs & \chk & \crs \\
Pub                                 & 2012 & Dart                & \bcircle{P}                               & \texttt{pubspec.yaml}           & \chk & \crs & \crs & \crs & \crs & \crs & \crs & \crs & \crs \\
DNF\tfnref{dnf-yum}
                                    & 2013 & Red Hat-based Linux & \bcircle{P}                              & \texttt{.spec}                  & \chk & \chk & \crs & \crs & \crs & \chk & \crs & \chk & \crs \\
opam~\cite{opam}                    & 2013 & OCaml               & \bcircle{P}\tfnref{opam-langagnostic}     & \texttt{.opam}                  & \chk & \chk & \crs & \crs & \crs & \chk & \chk & \crs & \chk \\
pkgman                              & 2013 & Haiku OS            & \bcircle{P}                               & \texttt{.PackageInfo}           & \chk & \chk & \crs & \crs & \crs & \crs & \crs & \chk & \crs \\
Spack~\cite{spack}                  & 2014 & HPC\tfnref{spack-hpc}
                                                                 & \bcircle{P}                               & Python                          & \chk & \chk & \chk & \crs & \chk & \crs & \chk & \chk & \crs \\
Cargo~\cite{cargo}                  & 2014 & Rust                & \bcircle{P} \bcircle{B}                   & \texttt{Cargo.toml}             & \chk & \crs & \chk & \crs & \chk & \crs & \chk & \crs & \crs \\
FuseSoC                             & 2014 & HDL\tfnref{fusesoc-hdl}
                                                                 & \bcircle{P} \bcircle{B}                   & \texttt{.core} (YAML)           & \chk & \crs & \crs & \crs & \crs & \crs & \chk & \chk & \crs \\
Bazel                               & 2015 & Multi-language      & \bcircle{P} \bcircle{B}\tfnref{bazel-build-system} & Starlark               & MVS  & \crs & \chk & \crs & \crs & \crs & \crs & \crs & \crs \\
Go modules~\cite{gomodules}         & 2018 & Go                  & \bcircle{P} \bcircle{B} \bcircle{C}       & \texttt{go.mod}                 & MVS  & \crs & \chk & \crs & \crs & \crs & \crs & \crs & \crs \\
Unison~\cite{unison}                & 2019 & Unison              & \bcircle{P} \bcircle{B} \bcircle{C}       & Unison                          & \crs & \crs & \chk & \crs & \crs & \crs & \crs & \crs & \crs \\
pack                                & 2022 & Idris 2             & \bcircle{P} \bcircle{B}                   & TOML + \texttt{.ipkg}           & \crs & \crs & \crs & \crs & \crs & \crs & \crs & \crs & \crs \\
\end{tabular}
\par\vspace{12pt}
\begin{multicols}{3}
\footnotesize\raggedright
\tfntext{toolchain-legend}{\bcircle{P}ackage manager, \bcircle{B}uild system, and/or \bcircle{C}ompiler.}
\tfntext{cpan-greedy}{CPAN installs the latest version greedily without backtracking.}
\tfntext{portage-slotting}{Limited support with `Slotting'.}
\tfntext{nix-singular}{Nix uses singular dependencies~(Appendix~\ref{appendix.singular-dependencies}) and resolves no version constraints, so most columns do not apply.}
\tfntext{maven-nearest-wins}{Maven uses a `nearest wins' heuristic over constraint solving.}
\tfntext{pip-poetry-uv}{Poetry and uv use \texttt{pyproject.toml}'s format and dependencies.}
\tfntext{homebrew-manual}{Emulated by hand with per-version formulae such as \texttt{python@3.11} -- the name mangling of \S\ref{sec.mise-en-place.concurrent} -- for ${\sim}2\%$ of packages.}
\tfntext{nuget-lowest}{NuGet defaults to lowest version without backtracking.}
\tfntext{chocolatey-nuget}{Chocolatey uses NuGet's resolution infrastructure.}
\tfntext{freebsd-pkg-ng}{\texttt{pkg} (pkgng) replaced the earlier FreeBSD \texttt{pkg\_*} tools.}
\tfntext{dnf-yum}{DNF replaced YUM as the standard package manager frontend for RPM-based distributions including Fedora.}
\tfntext{opam-langagnostic}{opam supports language-agnostic build scripts.}
\tfntext{spack-hpc}{High-performance computing.}
\tfntext{fusesoc-hdl}{Hardware description languages, such as Verilog and VHDL.}
\tfntext{bazel-build-system}{Bazel is a build system with package management functionality.}
\end{multicols}
\end{table}
\end{landscape}

\renewcommand{\footnotesize}{\oldfootnotesize}

}